\documentclass[journal]{IEEEtran}

\usepackage{amsfonts}
\usepackage{amssymb}
\usepackage{amsthm}
\usepackage{amsmath,amsfonts,amssymb}
\usepackage[dvips]{graphicx}
\usepackage{subfigure}
\usepackage{verbatim}
\usepackage{setspace}
\usepackage{bm}
\usepackage{algorithmic} 
\usepackage[ruled,vlined]{algorithm2e}
\usepackage{cite}

\usepackage{changepage}
\usepackage{pdfpages}
\usepackage{color}
\newtheorem{theorem}{Theorem}
\newtheorem{lemma}{Lemma}
\newcommand*{\dif}{\mathop{}\!\mathrm{d}}
\allowdisplaybreaks

\newcommand{\figwidth}{8.0}
\IEEEoverridecommandlockouts

\begin{document}
\title{Performance Analysis and Optimization for Movable Antenna Aided Wideband Communications}
\author{Lipeng Zhu, ~\IEEEmembership{Member,~IEEE,}
		Wenyan Ma,~\IEEEmembership{Graduate Student Member,~IEEE,}
		Zhenyu Xiao,~\IEEEmembership{Senior Member,~IEEE,}
		and Rui Zhang,~\IEEEmembership{Fellow,~IEEE}
	\vspace{-0.5 cm}
	\thanks{L. Zhu and W. Ma are with the Department of Electrical and Computer Engineering, National University of Singapore, Singapore 117583 (e-mail: zhulp@nus.edu.sg, wenyan@u.nus.edu).}
	\thanks{Z. Xiao is with the School of Electronic and Information Engineering, Beihang University, Beijing, China 100191. (e-mail: xiaozy@buaa.edu.cn)}
	\thanks{R. Zhang is with School of Science and Engineering, Shenzhen Research Institute of Big Data, The Chinese University of Hong Kong, Shenzhen, Guangdong 518172, China (e-mail: rzhang@cuhk.edu.cn). He is also with the Department of Electrical and Computer Engineering, National University of Singapore, Singapore 117583 (e-mail: elezhang@nus.edu.sg). }
}

\maketitle


\begin{abstract}
	Movable antenna (MA) has emerged as a promising technology to enhance wireless communication performance by enabling the local movement of antennas at the transmitter (Tx) and/or receiver (Rx) for achieving more favorable channel conditions. As the existing studies on MA-aided wireless communications have mainly considered narrow-band transmission in flat fading channels, we investigate in this paper the MA-aided wideband communications employing orthogonal frequency division multiplexing (OFDM) in frequency-selective fading channels. Under the general multi-tap field-response channel model, the wireless channel variations in both space and frequency are characterized with different positions of the MAs. Unlike the narrow-band transmission where the optimal MA position at the Tx/Rx simply maximizes the single-tap channel amplitude, the MA position in the wideband case needs to balance the amplitudes and phases over multiple channel taps in order to maximize the OFDM transmission rate over multiple frequency subcarriers. First, we derive an upper bound on the OFDM achievable rate in closed form when the size of the Tx/Rx region for antenna movement is arbitrarily large. Next, we develop a parallel greedy ascent (PGA) algorithm to obtain locally optimal solutions to the MAs' positions for OFDM rate maximization subject to finite-size Tx/Rx regions. To reduce computational complexity, a simplified PGA algorithm is also provided to optimize the MAs' positions more efficiently. Simulation results demonstrate that the proposed PGA algorithms can approach the OFDM rate upper bound closely with the increase of Tx/Rx region sizes and outperform conventional systems with fixed-position antennas (FPAs) under the wideband channel setup.
\end{abstract}

%
\IEEEpeerreviewmaketitle

\section{Introduction}
\IEEEPARstart{T}{he} future wireless communication systems are anticipated to achieve ultra-high transmission rates and reliability. Due to the scarcity of spectrum resources, numerous research endeavors have been devoted to enhancing the spectral efficiency in wireless communications. In this context, multi-antenna or multiple-input multiple-output (MIMO) \cite{Paulraj2004Anover,Stuber2004broadb,Larsson2014massive,Wan2021HoloRIS} has been recognized as a revolutionary technology for exploring the new degrees of freedom (DoFs) in the spatial domain. By exploiting the spatial multiplexing and diversity gains, MIMO technologies have substantially enhanced the transmission rate and reliability for wireless communication systems. However, the full utilization of spatial DoFs in conventional MIMO systems is hindered by discrete and fixed-position deployment of antennas at the transmitter (Tx) and receiver (Rx).

\subsection{Overview of Movable Antenna (MA)}
To overcome the fundamental limitation of conventional fixed-position antennas (FPAs), the flexible-position antenna has recently been regarded as a promising technology to enhance the MIMO communication performance, which is known as the fluid antenna system (FAS) \cite{wong2020limit,Wong2021fluid,wong2022bruce} or MA system  \cite{zhu2022MAmodel,ma2022MAmimo,zhu2023MAMag}\footnote{In fact, both fluid antenna and movable antenna have their longstanding presence in the field of antenna technology \cite{Tam2008patent,balanis2008mems} and were recently introduced for investigation in wireless communications, whereas both FAS and MA system are interchangeable terms in recognizing the potential of antenna position flexibility and are not limited to any particular way of implementation. A recent article \cite{zhu2024historical} attempted to clarify on the origins of the two terminologies.}. With the capability of flexible positioning/movement, the MAs can be positioned at locations yielding improved channel conditions in the continuous Tx/Rx regions. Thus, MA-aided wireless communication systems can exploit the full DoFs in the spatial domain, as compared to conventional FPA systems with channels suffering from random fading \cite{goldsmith2005wireless,Molisch2004MIMOsys,Sanayei2004antennas}. By leveraging the spatial diversity, the MA position optimization can boost the desired signal power as well as suppress the undesired interference \cite{zhu2023MAMag,zhu2022MAmodel}. Moreover, the spatial multiplexing performance can be enhanced for MA-aided MIMO or multi-user communication systems by reshaping the MIMO channel matrices via MAs' joint position optimization \cite{ma2022MAmimo,zhu2023MAmultiuser}.

Note that the size of space for antenna movement is generally in the order of signal wavelength \cite{zhu2023MAMag,zhu2022MAmodel,ma2022MAmimo}, which eases the implementation of MA systems, especially in the high frequency bands with small wavelengths \cite{zhu2019millim,Ning2023THzbeam}. Specifically, an architecture of MA-mounted Tx/Rx was presented in \cite{zhu2023MAMag}, which comprises a conventional communication module and an antenna positioning module controlled by a central unit. The MA can be flexibly moved in a three-dimensional (3D) space with the aid of mechanical slides and step motors for improving the communication performance. Besides, a prototype of MA systems was developed in \cite{Zhuravlev2015experi} for radar applications, where the Tx and Rx antennas can be moved over line segments with the aid of motor-based drivers. Moreover, the authors in \cite{Basbug2017design} designed an MA array for synthesizing flexible beamforming patterns, where each antenna element can be locally moved along a semicircular trajectory by step motors. Note that the motor-enabled MA architectures usually require an extended area for installing slides and motors, which may not be applicable to devices with a small size. In contrast, the micro-electromechanical systems (MEMS)-enabled MA \cite{balanis2008mems} has the advantage of size miniaturization, high positioning accuracy, and low power consumption, which is more suitable to be integrated in compact devices. In addition, the liquid-based antenna and the pixel-based antenna presented in \cite{wong2022bruce} are two alternative ways for implementing the FAS/MA with fast adaptation in antenna position.

\subsection{Related Works}

The exploration of MA-aided communications started from the perspective of point-to-point transmissions \cite{zhu2022MAmodel,ma2022MAmimo,Do2021reconf,Do2022TeMIMO,chen2023joint,wong2020limit,Wong2021fluid,New2024fluid}. In \cite{zhu2022MAmodel}, the field-response based channel model was proposed for MA-enabled communication systems, aimed at characterizing the continuous variations of wireless multipath channels in both the Tx and Rx regions. By leveraging this channel model, the signal-to-noise ratio (SNR) enhancement of an MA system over its FPA counterpart was analyzed. The results revealed that an increased number of channel paths and an enlarged region for antenna movement can yield a more significant SNR improvement. In \cite{ma2022MAmimo}, the channel capacity of MA-aided MIMO systems was characterized by considering the additional DoFs in optimizing multiple MAs' positions at both the Tx and Rx jointly. It was shown that the MA position optimization can not only increase the total channel power but also balance the singular values of the resulted MIMO channel matrix flexibly such that the MA-MIMO channel capacity is maximized. The authors in \cite{Do2021reconf} and \cite{Do2022TeMIMO} proposed the line-of-sight (LoS) MIMO transmission empowered by a rotational uniform linear array (ULA), which can be regarded as a special way of implementing MAs in confined regions. It was revealed that the upper bound on the capacity of LoS MIMO systems can be asymptotically approached by rotating the Tx/Rx ULA with an SNR-dependent angle. Moreover, the authors in \cite{chen2023joint} investigated the joint beamforming and antenna movement deign for MA-enhanced MIMO systems based on statistical channel state information (CSI) between the Tx and Rx regions. A constrained stochastic successive convex approximation (CSSCA) algorithm was developed to maximize the ergodic capacity, which can reap 20\% improvement compared to conventional FPA-MIMO systems. In addition, the spatial correlation channel model was adopted in \cite{wong2020limit,Wong2021fluid,New2024fluid} to characterize the outage probability and ergodic capacity of FASs, which demonstrated their superior performance compared to FPA systems with maximum ratio combining.

The superiority of MAs over FPAs has also been validated in multiuser communication systems \cite{zhu2023MAmultiuser,hu2023power,xiao2023multiuser,cheng2023sum,qin2023antenna,wu2023movable,Wong2022fluid,wong2023slowfluid,wong2023fastfluid,Wong2023opport,zhu2023MAarray,Leshem2021align,ma2023multi,hu2023secure,hu2023comp}. In \cite{zhu2023MAmultiuser}, the MAs were employed at the users' side to improve the multiple access channel (MAC) capacity. By jointly optimizing the MA's position and transmit power at each user as well as the receive combining matrix at the base station (BS), the total transmit power of multiple users can be significantly decreased for meeting a given rate requirement of each user. This is due to the MAs' position optimization that can help reduce the correlation between users' channel vectors and thus alleviate the multiuser interference. Under this setup, the authors in \cite{hu2023power} developed a projected gradient descent-based algorithm to further reduce the computational complexity. In \cite{xiao2023multiuser}, the MA-enabled BS was investigated to enhance the user fairness performance in the uplink, where a particle swarm optimization (PSO)-based algorithm was developed to maximize the minimum achievable rate among multiple users, subject to a constraint on minimum inter-MA distance at the BS. Besides, the sum-rate maximization and transmit-power minimization problems for downlink transmission between the MA-enabled BS and multiple users were investigated in \cite{cheng2023sum} and \cite{qin2023antenna}, respectively. A variety of optimization techniques have been adopted to obtain suboptimal solutions for the MAs’ positions and beamforming matrix at the BS, e.g., fractional programming, alternating optimization, gradient descent, penalty method, and successive convex approximation (SCA). The authors in \cite{wu2023movable} considered the discrete antenna positioning and beamforming for MA-enhanced multiuser communication systems. To minimize the total transmit power while guaranteeing the minimum rate requirement of each user, the optimal solution was obtained by an iterative algorithm based on the generalized Bender’s decomposition. Following the principle of discrete antenna port selection, the outage probability and spatial multiplexing gain were characterized in \cite{Wong2022fluid,wong2023slowfluid,wong2023fastfluid,Wong2023opport} under different setups of multiple access aided by the FAS. Moreover, it was revealed in \cite{zhu2023MAarray} that for MA array-enhanced beamforming, the full array gain over the direction of desired signals and the interference nulling over undesired directions can be simultaneously achieved, where the optimal antenna positioning and beamforming vectors were derived in closed form. In \cite{Leshem2021align}, the authors demonstrated that by only adjusting the distance of adjacent antennas in a sufficiently large region, the interference from an arbitrary large number of spatial directions can be nulled to any desired level, while the interference-free SNR is maintained. In addition, the investigations in \cite{ma2023multi,hu2023secure,hu2023comp} substantiated the superiority of MA arrays in improving the performance of multi-beam forming, secure communication, and coordinated multi-point (CoMP) transmission.

Note that the performance gain of MAs requires the knowledge of accurate CSI from any position in the Tx region to any position in the Rx region. To this end, a successive transmitter-receiver compressed sensing (STRCS) channel estimation method was proposed in \cite{ma2023MAestimation} to sequentially resolve the angles of departure (AoDs), the angles of arrival (AoAs), and the complex coefficients for multiple channel paths, based on a finite number of channel measurements at designated MA locations in the Tx and Rx regions. To avoid cumulative errors caused by sequential estimation, a joint AoD, AoA, and path coefficient estimation framework was proposed in \cite{xiao2023channel} based on the compressed sensing theory, where the criteria for MA movement/measurement positions were provided to guarantee the successful recovery of channel paths in the angular domain. Moreover, a successive Bayesian reconstructor (S-BAR) was proposed in \cite{zhang2023successive} to estimate the channel response from an FPA at the Tx to all candidate positions/ports of MAs at the Rx. This approach models the channel as a stochastic process and successively eliminates the channel uncertainty by kernel-based sampling and regression at different locations of MAs.

\subsection{Motivation and Contribution}
It is worth pointing out that the aforementioned studies \cite{zhu2022MAmodel,ma2022MAmimo,Do2021reconf,Do2022TeMIMO,chen2023joint,wong2020limit,Wong2021fluid,New2024fluid,zhu2023MAmultiuser,hu2023power,xiao2023multiuser,cheng2023sum,qin2023antenna,wu2023movable,Wong2022fluid,wong2023slowfluid,wong2023fastfluid,Wong2023opport,zhu2023MAarray,Leshem2021align,ma2023multi,hu2023secure,hu2023comp,ma2023MAestimation,xiao2023channel,zhang2023successive} mainly concentrate on the design of MA-enabled systems for narrow-band communications under flat fading channels. However, there is very limited work exploring the potential performance enhancement by MAs in the general wideband communications under frequency-selective fading channels. To fill this gap, we investigate in this paper the channel modeling, performance analysis, and performance optimization for MA-aided wideband communication systems. The main contributions of this paper are summarized as follows:
\begin{itemize}
	\item We consider an MA-aided orthogonal frequency division multiplexing (OFDM) wideband communication system, where the Tx and Rx are each equipped with an MA which can be moved continuously in a 3D region. A general multi-tap field-response channel model is adopted to characterize the wireless channel variations in both space and frequency with respect to (w.r.t.) different positions of the MAs at the Tx and Rx sides. Then, an optimization problem is formulated to maximize the OFDM achievable rate by jointly optimizing the positions for MAs and the transmit power allocation over OFDM subcarriers.
	\item Next, analytical results are provided to unveil the asymptotic performance of the MA-OFDM communication system. Specifically, we demonstrate the great potential of MA positioning for achieving the desired channel impulse response (CIR) with maximum channel gains yet arbitrary channel phases over all clustered delay taps. Based on this finding, an upper bound on the OFDM achievable rate is derived in closed form in the high-SNR regime when the size of the Tx/Rx region for antenna movement is arbitrarily large.
	\item Furthermore, under the practical constraint of finite-size Tx/Rx regions, we develop a parallel greedy ascent (PGA) algorithm for optimizing the MA positioning vectors, with the optimal power allocation given by the water-filling criterion. To reduce computational complexity, a simplified PGA algorithm is also provided to optimize the MAs' positions more efficiently by maximizing the total channel power gain over all delay taps, where the optimization of transmit power allocation is not needed over the iterations.
	\item Finally, extensive simulation results are presented to evaluate the performance of the proposed MA-OFDM wideband communication system design and optimization. It is shown that the proposed PGA algorithms can approach the OFDM rate upper bound closely with the increase of Tx/Rx region sizes and outperform conventional systems with FPAs under the wideband channel setup. It is also revealed that the proposed MA-OFDM system can yield more significant performance gains over its FPA counterpart under wireless channels with a small number of clustered delay taps each encompassing a large number of independent paths.
\end{itemize}

\subsection{Organization and Notation}
The rest of this paper is organized as follows. In Section II, we introduce the system model and the multi-tap field-response channel model for the MA-OFDM communication system. In Section III, we show the main analytical results for MA-OFDM communication systems. In Section IV, we develop the PGA algorithm and its simplified version for maximizing the achievable rate. Simulation results are provided in Section V and this paper is finally concluded in Section VI.

\textit{Notation}: $a$, $\mathbf{a}$, $\mathbf{A}$, and $\mathcal{A}$ denote a scalar, a vector, a matrix, and a set, respectively. $(\cdot)^{\rm{T}}$, $(\cdot)^{*}$, and $(\cdot)^{\rm{H}}$ denote transpose, conjugate, and conjugate transpose, respectively. $\mathcal{A} \setminus \mathcal{B}$ and $\mathcal{A} \cup \mathcal{B}$ represent the subtraction set and union set of $\mathcal{A}$ and $\mathcal{B}$, respectively. $|\mathcal{A}|$ and $\mathcal{L}(\mathcal{A})$ represent the cardinality and Lebesgue measure of set $\mathcal{A}$, respectively. $\mathcal{CN}(0,\sigma^{2})$ denotes the circularly symmetric complex Gaussian (CSCG) distribution with mean zero and covariance $\sigma^{2}$. $\Pr\{\cdot\}$ denotes the probability of an event.  $\mathbb{Z}$, $\mathbb{Q}$, $\mathbb{R}$, and $\mathbb{C}$ represent the sets of integer, rational, real, and complex numbers, respectively. $|\cdot|$ and $\angle(\cdot)$ denote the amplitude and the phase of a complex number or complex vector, respectively. $\|\cdot\|_{1}$ and $\|\cdot\|_{2}$ denote the 1-norm and 2-norm of a vector, respectively. $(\mathbf{a} \mod b)$ is the modulo operation of each element in $\mathbf{a}$ divided by integer $b$. $\mathrm{diag}\{\mathbf{a}\}$ is a diagonal matrix with the element in row $i$ and column $i$ equal to the $i$-th element of vector $\mathbf{a}$. $\partial(\cdot)$ denotes the partial differential of a function. $\nabla_{\mathbf{x}} f(\mathbf{x})$ represent the gradient of $f(\mathbf{x})$ w.r.t. $\mathbf{x}$. $\mathbf{0}_{L}$ denotes an $L$-dimensional row vector with all elements equal to 0. $\mathbf{I}_{L}$ denotes an identical matrix of size $L \times L$.

\begin{figure*}[t]
	\begin{center}
		\includegraphics[width=14 cm]{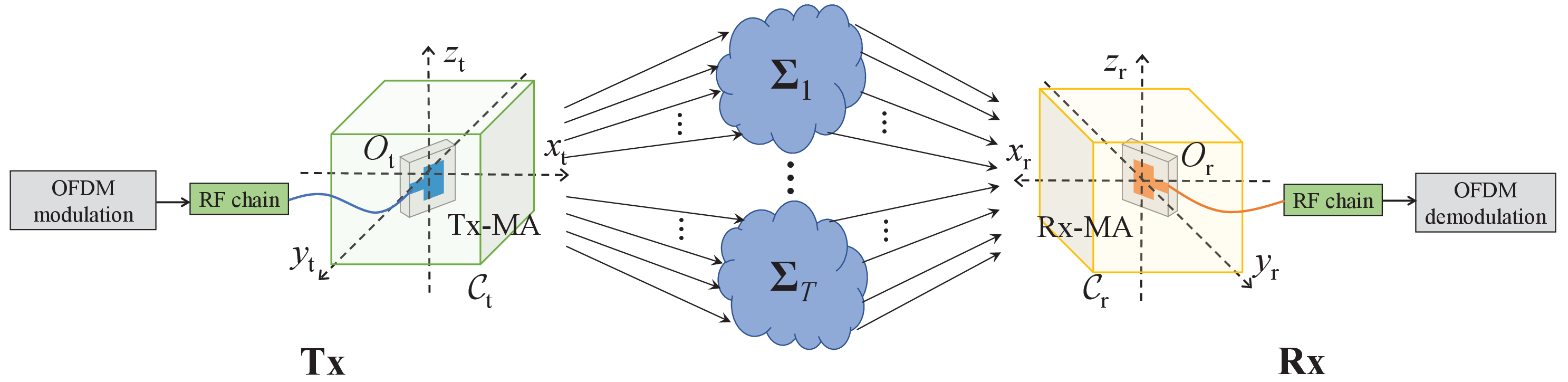}
		\caption{Illustration of the considered MA-OFDM wideband communication system.}
		\label{fig:system}
	\end{center}
\end{figure*}

\section{System Model}
As shown in Fig. 1, the Tx and Rx are each equipped with a single MA to enhance their wideband communication performance. The local 3D coordinates of the Tx-MA and the Rx-MA are denoted as $\mathbf{t}=[x_{\mathrm{t}}, y_{\mathrm{t}}, z_{\mathrm{t}}]^{\mathrm{T}}$ and $\mathbf{r}=[x_{\mathrm{r}}, y_{\mathrm{r}}, z_{\mathrm{r}}]^{\mathrm{T}}$, respectively. With the aid of driver components, the Tx-MA and Rx-MA can be moved in 3D regions $\mathcal{C}_{\mathrm{t}}$ and $\mathcal{C}_{\mathrm{r}}$, respectively. Without loss of generality, we assume in this paper that the regions for antenna moving are cuboids, i.e., $\mathcal{C}_{\mathrm{t}}=[x_{\mathrm{t}}^{\min}, x_{\mathrm{t}}^{\max}] \times [y_{\mathrm{t}}^{\min}, y_{\mathrm{t}}^{\max}] \times [z_{\mathrm{t}}^{\min}, z_{\mathrm{t}}^{\max}]$ and $\mathcal{C}_{\mathrm{r}}=[x_{\mathrm{r}}^{\min}, x_{\mathrm{r}}^{\max}] \times [y_{\mathrm{r}}^{\min}, y_{\mathrm{r}}^{\max}] \times [z_{\mathrm{r}}^{\min}, z_{\mathrm{r}}^{\max}]$.

\subsection{Channel Model}
For the considered MA-OFDM system, denote $B$ as the system bandwidth and $M$ as the total number of subcarriers. Given the finite bandwidth, the baseband equivalent channel can be characterized by multiple delay taps, each spanning a time interval of $1/B$. In particular, we denote the maximum number of delay taps for the baseband equivalent channel as $T$, and thus the length of cyclic prefix (CP) should be $M_{\mathrm{CP}} \geq T$. Since the distance between the Tx and Rx is generally much larger than the size of antenna-moving regions, the far-field condition holds between the Tx-MA and Rx-MA. For example, if the carrier frequency is 5.2 GHz and the size of Tx/Rx regions for antenna movement is 5 wavelengths, the corresponding Rayleigh distance is no larger than 3 meters \cite{kraus2002antennas}, which is easily surpassed by the practical Tx-Rx distance such that the far-field propagation condition is guaranteed. Thus, the plane-wave model can be used to characterize the multi-tap multi-path channel between the transceivers. In other words, the AoD, the AoA, and the amplitude of complex coefficient for each channel path between the Tx and Rx regions are invariant, while only the phase of each path's coefficient changes with the positions of MAs \cite{zhu2023MAMag,zhu2022MAmodel}. 

\begin{figure}[t]
	\begin{center}
		\includegraphics[width=8.8 cm]{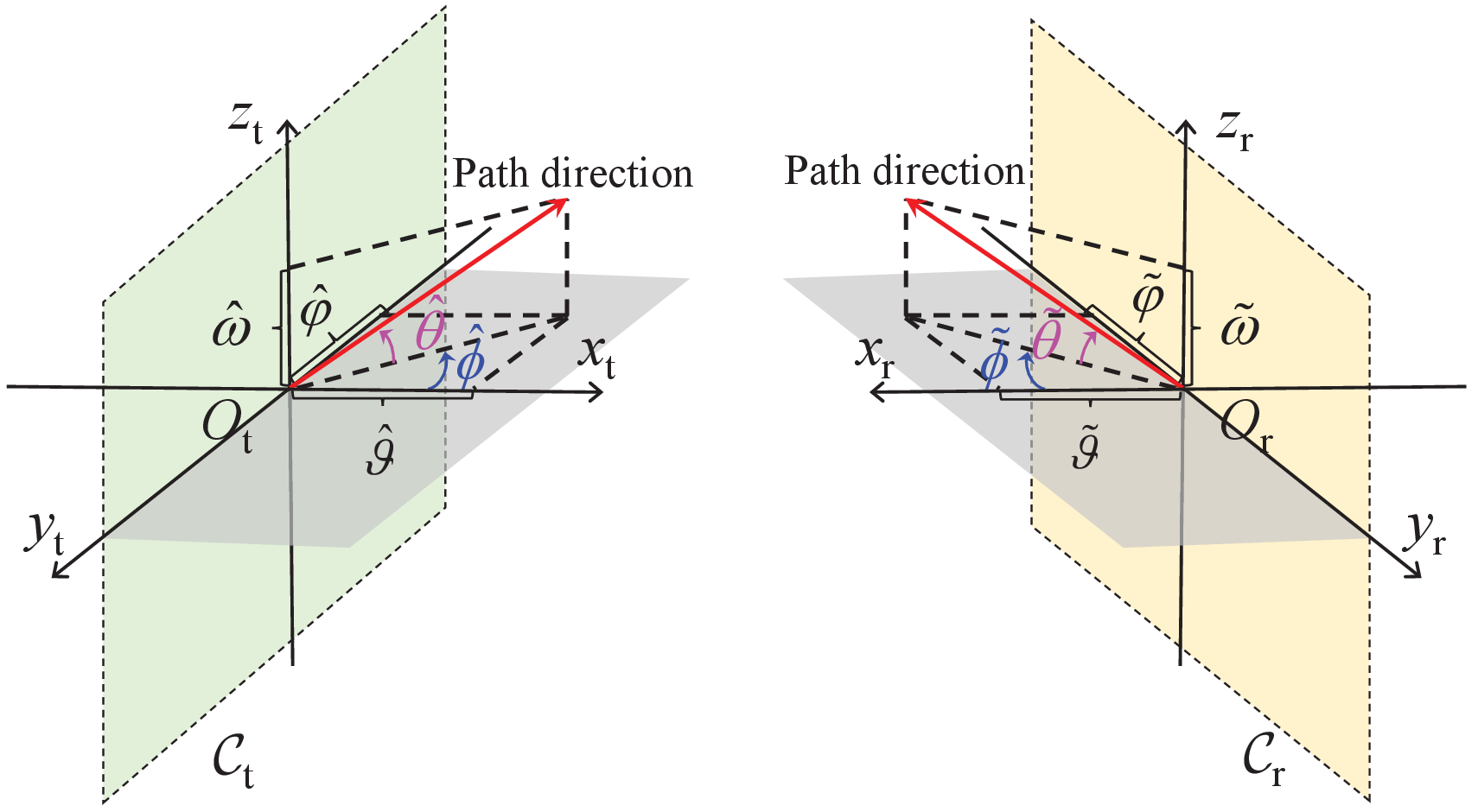}
		\caption{Illustration of the local 3D coordinate systems at the Tx/Rx and the corresponding AoDs/AoAs.}
		\label{fig:Coordinates}
	\end{center}
\end{figure}

For the $\tau$-th delay tap, we denote the number of (clustered) channel paths as $L_{\tau}$, $1 \leq \tau \leq T$. As shown in Fig. 2, the elevation and azimuth AoDs for the $\ell$-th channel path over the $\tau$-th delay tap between the Tx and Rx are denoted by $\hat{\theta}_{\tau}^{\ell}$ and $\hat{\phi}_{\tau}^{\ell}$, respectively, $1 \leq \ell \leq L_{\tau}$. The elevation and azimuth AoAs for the $\ell$-th channel path over the $\tau$-th delay tap between the Tx and Rx are denoted by $\tilde{\theta}_{\tau}^{\ell}$ and $\tilde{\phi}_{\tau}^{\ell}$, respectively, $1 \leq \ell \leq L_{\tau}$. For convenience, the virtual AoDs and AoAs are defined as
\begin{subequations}\label{eq_AoD_AoA}
	\begin{align}
		&\begin{aligned}&\hat{\vartheta}_{\tau}^{\ell} = \cos \hat{\theta}_{\tau}^{\ell} \cos \hat{\phi}_{\tau}^{\ell},~
			\hat{\varphi}_{\tau}^{\ell} = \cos \hat{\theta}_{\tau}^{\ell} \sin \hat{\phi}_{\tau}^{\ell},~
			\hat{\omega}_{\tau}^{\ell} = \sin \hat{\theta}_{\tau}^{\ell}, \end{aligned} \label{eq_AoD}\\
		&\begin{aligned}&\tilde{\vartheta}_{\tau}^{\ell} = \cos \tilde{\theta}_{\tau}^{\ell} \cos \tilde{\phi}_{\tau}^{\ell},~
			\tilde{\varphi}_{\tau}^{\ell} = \cos \tilde{\theta}_{\tau}^{\ell} \sin \tilde{\phi}_{\tau}^{\ell},~
			\tilde{\omega}_{\tau}^{\ell} = \sin \tilde{\theta}_{\tau}^{\ell}, \end{aligned}  \label{eq_AoA}
	\end{align}
\end{subequations} 
for $1 \leq \ell \leq L_{\tau}$ and $1 \leq \tau \leq T$. Moreover, the wave vectors at the Tx and Rx are defined as $\hat{\mathbf{k}}_{\tau}^{\ell} = [\hat{\vartheta}_{\tau}^{\ell}, \hat{\varphi}_{\tau}^{\ell}, \hat{\omega}_{\tau}^{\ell}]^{\mathrm{T}}$ and $\tilde{\mathbf{k}}_{\tau}^{\ell} = [\tilde{\vartheta}_{\tau}^{\ell}, \tilde{\varphi}_{\tau}^{\ell}, \tilde{\omega}_{\tau}^{\ell}]^{\mathrm{T}}$, respectively. Then, the field-response vectors (FRVs) for the channel paths over the $\tau$-th (clustered) delay tap, $1 \leq \tau \leq T$, between the Tx and Rx are given by \cite{zhu2023MAMag,zhu2022MAmodel}
\begin{subequations}\label{eq_field_res}
	\begin{align}
		&\mathbf{g}_{\tau}(\mathbf{t}) = \left[e^{j\frac{2\pi}{\lambda} \mathbf{t}^{\mathrm{T}} \hat{\mathbf{k}}_{\tau}^{1}},
		e^{j\frac{2\pi}{\lambda} \mathbf{t}^{\mathrm{T}} \hat{\mathbf{k}}_{\tau}^{2}},
		\cdots, e^{j\frac{2\pi}{\lambda} \mathbf{t}^{\mathrm{T}} \hat{\mathbf{k}}_{\tau}^{L_{\tau}}}\right]^{\mathrm{T}}, \label{eq_field_resTx}\\
		&\mathbf{f}_{\tau}(\mathbf{r}) = \left[e^{j\frac{2\pi}{\lambda} \mathbf{r}^{\mathrm{T}} \tilde{\mathbf{k}}_{\tau}^{1}},
		e^{j\frac{2\pi}{\lambda} \mathbf{r}^{\mathrm{T}} \tilde{\mathbf{k}}_{\tau}^{2}},
		\cdots, e^{j\frac{2\pi}{\lambda} \mathbf{r}^{\mathrm{T}} \tilde{\mathbf{k}}_{\tau}^{L_{\tau}}}\right]^{\mathrm{T}}. \label{eq_field_resRx}
	\end{align}
\end{subequations}
In particular, $\mathbf{t}^{\mathrm{T}} \hat{\mathbf{k}}_{\tau}^{\ell}$, $1 \leq \ell \leq L_{\tau}$, characterizes the difference of the signal propagation distance for the $\ell$-th channel path between Tx-MA position $\mathbf{t}$ and the reference point of the Tx region. Similarly, $\mathbf{r}^{\mathrm{T}} \tilde{\mathbf{k}}_{\tau}^{\ell}$, $1 \leq \ell \leq L_{\tau}$, represents the difference of the signal propagation distance for the $\ell$-th channel path between Rx-MA position $\mathbf{r}$ and the reference point of the Rx region. Thus, the FRVs account for the phase changes of the complex coefficients for all channel paths under different positions of the Tx-MA and Rx-MA.

As such, the baseband equivalent (time-domain) CIR over the $\tau$-th delay tap between the Tx-MA located at position $\mathbf{t}$ and the Rx-MA located at position $\mathbf{r}$ can be represented as
\begin{equation}\label{eq_CIR}
	h_{\tau}(\mathbf{t},\mathbf{r}) = \mathbf{f}_{\tau}(\mathbf{r})^{\mathrm{H}} \mathbf{\Sigma}_{\tau} \mathbf{g}_{\tau}(\mathbf{t}),~1 \leq \tau \leq T,
\end{equation}
where $\mathbf{\Sigma}_{\tau} = \mathrm{diag}\{\mathbf{b}_{\tau}\} \in \mathbb{C}^{L_{\tau} \times L_{\tau}}$ represents the path-response matrix (PRM) and $\mathbf{b}_{\tau}=[b_{\tau}^{1}, b_{\tau}^{2}, \cdots, b_{\tau}^{L_{\tau}}]^{\mathrm{T}}$ includes the response coefficients of all $L_{\tau}$ channel paths from the reference point of the Tx region to the reference point of the Rx region. 
Let $\mathbf{h}(\mathbf{t},\mathbf{r})=[h_{1}(\mathbf{t},\mathbf{r}), \cdots, h_{T} (\mathbf{t},\mathbf{r}), \mathbf{0}_{M-T}]^{\mathrm{T}} \in \mathbb{C}^{M}$ denote the zero-padded baseband equivalent CIR vector. The channel frequency response (CFR) over all the subcarriers between the Tx-MA located at position $\mathbf{t}$ and the Rx-MA located at position $\mathbf{r}$ is thus given by \cite{goldsmith2005wireless}
\begin{equation}\label{eq_CFR}
	\mathbf{c}(\mathbf{t},\mathbf{r}) = \mathbf{D}_{M} \mathbf{h}(\mathbf{t},\mathbf{r}) \triangleq [c_{1}(\mathbf{t},\mathbf{r}), c_{2}(\mathbf{t},\mathbf{r}), \cdots, c_{M}(\mathbf{t},\mathbf{r})]^{\mathrm{T}},
\end{equation}
where $\mathbf{D}_{M}$ denotes the $M$-dimensional discrete Fourier transform (DFT) matrix.

As can be observed from \eqref{eq_CIR} and \eqref{eq_CFR}, the change of MAs' positions can yield different combinations of the entries in the PRM over each delay tap. Thus, by optimizing the positions of the Tx-MA and Rx-MA, the CIR and CFR of the considered MA-OFDM system can be significantly reconfigured. For example, if the Tx-MA and Rx-MA are deployed at proper positions such that the complex coefficients of channel paths over each delay tap are constructively superimposed, the average channel power gain over $M$ subcarriers can be improved. In contrast, if the Tx-MA and Rx-MA are placed at positions where the complex coefficients of channel paths over each delay tap are destructively superimposed, the average channel power decreases. In a word, the incorporation of additional DoFs in position optimization of the Tx-MA and Rx-MA facilitates enhancing the channel conditions and thereby improving communication performance of the considered MA-OFDM system.

\subsection{Problem Formulation}
Denote the transmit power allocation vector as $\mathbf{p}=[p_{1}, p_{2}, \cdots, p_{M}]^{\mathrm{T}} \in \mathbb{R}^{M}$, where $p_{m} \geq 0$ represents the power allocated to the $m$-th subcarrier, $1 \leq m \leq M$. Then, the achievable rate for the considered MA-OFDM system is given by \cite{goldsmith2005wireless}
\begin{equation}\label{eq_rate}
	R(\mathbf{t},\mathbf{r},\mathbf{p}) = \frac{1}{M+M_{\mathrm{CP}}} \sum \limits_{m=1}^{M} \log_{2} \left(1+\frac{|c_{m}(\mathbf{t},\mathbf{r})|^{2}p_{m}}{\sigma^{2}}\right),
\end{equation}
where $\sigma^{2}$ denotes the noise power for each subcarrier.

In this paper, we aim to maximize the achievable rate for the considered MA-OFDM system by exploiting the new DoF in antenna position optimization, which can be expressed as the following optimization problem:
\begin{subequations}\label{eq_problem}
	\begin{align}
		\mathop{\max}\limits_{\mathbf{t}, \mathbf{r}, \mathbf{p}}~~~
		&R(\mathbf{t},\mathbf{r},\mathbf{p}) \label{eq_problem_a}\\
		\mathrm{s.t.}~~~~ &p_{m} \geq 0,~1 \leq m \leq M, \label{eq_problem_b}\\
		&\sum \limits_{m=1}^{M} p_{m} \leq P, \label{eq_problem_c}\\
		&\mathbf{t} \in \mathcal{C}_{\mathrm{t}}, \label{eq_problem_d}\\
		&\mathbf{r} \in \mathcal{C}_{\mathrm{r}}, \label{eq_problem_e}
	\end{align}
\end{subequations}
where constraint \eqref{eq_problem_b} guarantees that the power allocated to each subcarrier is non-negative; constraint \eqref{eq_problem_c} ensures that the total transmit power does not exceed its maximum value $P$; and constraints \eqref{eq_problem_d} and \eqref{eq_problem_e} confine the Tx-MA and Rx-MA moving in their feasible regions, respectively. Since the achievable rate is highly non-linear w.r.t. variables $\mathbf{t}$ and $\mathbf{r}$, it is challenging to derive the optimal solution for problem \eqref{eq_problem}. In the following sections, we first analyze the asymptotic performance upper bound on the OFDM achievable rate. Then, the practical algorithms are developed to obtain suboptimal solutions for problem \eqref{eq_problem}.

\section{Performance Analysis}
To facilitate the performance analysis for MA-OFDM systems, we first characterize the property of virtual AoDs and AoAs (i.e., $\{\hat{\vartheta}_{\tau}^{\ell}\}$, $\{\hat{\varphi}_{\tau}^{\ell}\}$, $\{\hat{\omega}_{\tau}^{\ell}\}$, $\{\tilde{\vartheta}_{\tau}^{\ell}\}$, $\{\tilde{\varphi}_{\tau}^{\ell}\}$, and $\{\tilde{\omega}_{\tau}^{\ell}\}$) for the channel paths between the Tx and Rx because these parameters fundamentally determine the spatial diversity of wireless channels. Without loss of generality, we consider the virtual AoD vector, $\bm{\vartheta} \triangleq [\hat{\vartheta}_{1}^{1}, \cdots, \hat{\vartheta}_{1}^{L_{1}}, \cdots, \hat{\vartheta}_{T}^{1}, \cdots, \hat{\vartheta}_{T}^{L_{T}}]^{\mathrm{T}}$, of which the dimension is denoted by $N=\sum_{\tau=1}^{T} L_{\tau}$. In particular, we call that the elements in $\bm{\vartheta}$ are linearly independent over the rational number set $\mathbb{Q}$ if for any non-zero vector $[a_{1}^{1},\cdots,a_{1}^{L_{1}},\cdots,a_{T}^{1},\cdots,a_{T}^{L_{T}}]^{\mathrm{T}} \in \mathbb{Q}^{N}$, the inequation $\sum_{\tau=1}^{T} \sum_{\ell=1}^{L_{\tau}}  a_{\tau}^{\ell} \hat{\vartheta}_{\tau}^{\ell} \neq 0$
always holds, which is termed as the \emph{linearly independent angle (LIA)} condition\footnote{Note that the LIA condition and subsequent analysis are also applicable to all other virtual AoDs/AoAs, $\{\hat{\varphi}_{\tau}^{\ell}\}$, $\{\hat{\omega}_{\tau}^{\ell}\}$, $\{\tilde{\vartheta}_{\tau}^{\ell}\}$, $\{\tilde{\varphi}_{\tau}^{\ell}\}$, and $\{\tilde{\omega}_{\tau}^{\ell}\}$.}. 

In fact, if the virtual AoDs in $\bm{\vartheta}$ are linearly dependent over $\mathbb{Q}$, i.e., the LIA condition does not hold, then there always exists a period $X$ in distance which guarantees that $\mathbf{h}(\mathbf{t},\mathbf{r}) = \mathbf{h}(\mathbf{t}+[X,0,0]^{\mathrm{T}},\mathbf{r})$ holds for any $\mathbf{t}$ and $\mathbf{r}$ \cite{zhu2022MAmodel}. Such a periodic behavior decreases the spatial diversity of wireless channels because the CIRs have a high spatial correlation between any two periods. In contrast, if the virtual AoDs in $\bm{\vartheta}$ are linearly independent over $\mathbb{Q}$, i.e., the LIA condition holds, then there is no explicit period for the CIR and thus its maximal diversity can be achieved in the spatial domain.

Next, we will demonstrate that the virtual AoDs almost always satisfy the LIA condition in a probabilistic sense. In practice, due to the random locations of scatterers in the signal propagation environment, the virtual AoDs can be modeled as independent random variables within $[-1,1]$, which are denoted by a random vector, $\bm{\Theta} \triangleq [\Theta_{1}^{1}, \cdots, \Theta_{1}^{L_{1}}, \cdots, \Theta_{T}^{1}, \cdots, \Theta_{T}^{L_{T}}]^{\mathrm{T}}$. Denoting $f_{\bm{\Theta}}(\bm{\vartheta})=\prod_{\tau=1}^{T} \prod_{\ell=1}^{L_{\tau}} f_{\Theta_{\tau}^{\ell}}(\hat{\vartheta}_{\tau}^{\ell})$ as the joint probability density function (PDF) of random vector $\bm{\Theta}$, we have the following lemma to characterize the probability of the virtual AoDs satisfying the LIA condition.

\begin{lemma}\label{lemma_LIA}
	If PDFs $\{f_{\Theta_{\tau}^{\ell}}(\hat{\vartheta}_{\tau}^{\ell})\}$ are continuous and bounded, then we always have
	\begin{equation}
		\Pr\left\{\bm{\Theta} \in \mathcal{J}\right\} = \int_{\bm{\vartheta} \in \mathcal{J}} f_{\bm{\Theta}}\left(\bm{\vartheta}\right) \dif \bm{\vartheta} = 1,
	\end{equation}
	where $\mathcal{J}$ denotes the set of all vectors in $N$-dimensional interval $[-1,1]^{N}$ which satisfy the LIA condition.
\end{lemma}

\begin{proof}
	See Appendix \ref{Appen_LIA}.
\end{proof}

Lemma \ref{lemma_LIA} indicates that the virtual AoDs satisfy the LIA condition with probability 1, which motivates us to focus on the case of $\bm{\vartheta}$ satisfying the LIA condition. Moreover, to reveal the ultimate performance limit of MA-aided wideband communication systems, we assume that the region size for antenna moving along axis $x_{\mathrm{t}}$ is arbitrarily large, which is termed as the \emph{arbitrarily large region (ALR)} assumption. 
The following theorem demonstrates the great potential of MA positioning for achieving the desired CIR with maximum channel gains yet arbitrary channel phases over all clustered delay taps. 

\begin{theorem}\label{Theo_CIR}
	Under the LIA condition and the ALR assumption, for any small positive number $\delta \ll 1$ and any real (phase) value $\nu_{\tau}$, $ 1\leq \tau \leq T$, there always exist $\mathbf{t}$ and $\mathbf{r}$ satisfying
	\begin{equation}
		\left|\|\mathbf{b}_{\tau}\|_{1} e^{j2\pi\nu_{\tau}} - h_{\tau}(\mathbf{t}, \mathbf{r})\right| \leq \delta, ~1\leq \tau \leq T.
	\end{equation}
\end{theorem}

\begin{proof}
	See Appendix \ref{Appen_CIR}.
\end{proof}

Theorem \ref{Theo_CIR} indicates that the optimization of MA positioning can not only maximize the channel gains over all delay taps but also alter their channel phases flexibly, which can be designed according to practical communication requirements and channel conditions. For example, in the low-SNR regime, the best transmission strategy for maximizing the OFDM achievable rate is to allocate all transmit power to the subcarrier with the highest channel gain \cite{goldsmith2005wireless}. In such a case, the phase of the CIR vector should be designed aligning with that of a selected column in the DFT matrix such that the channel gain over the corresponding subcarrier is maximized, e.g., setting $\nu_{\tau}=0$, $1 \leq \tau \leq T$, for maximizing the channel gain over the first subcarrier\footnote{In the context of orthogonal frequency division multiple access (OFDMA) systems, multiple users experience independent frequency-selective fading channels and transmit/receive over different subcarriers. As such, each individual user can optimize the antenna position to achieve the maximum channel gain over its assigned subcarrier(s).}.

Next, we consider the high-SNR regime and derive an upper bound on the achievable rate for the considered MA-OFDM communication system by the following theorem. 
\begin{theorem}\label{Theo_rate}
	Under the LIA condition and the ALR assumption, the MA-OFDM achievable rate in the high-SNR regime is upper-bounded by
	\begin{equation}\label{eq_rate_bound}
		R(\mathbf{t},\mathbf{r},\mathbf{p}) \leq \bar{R} = \frac{M}{M+M_{\mathrm{CP}}} \log_{2} \left(1+\frac{GP}{M\sigma^{2}}\right),
	\end{equation}
	with $G=\sum_{\tau=1}^{T} \|\mathbf{b}_{\tau}\|_{1}^{2}$.
\end{theorem}

\begin{proof}
	See Appendix \ref{Appen_rate}.
\end{proof}
As can be observed from the proof of Theorem \ref{Theo_rate}, the upper bound on the achievable rate in \eqref{eq_rate_bound} is attained through equal transmit power allocation and equal channel gain realization across all OFDM subcarriers. However, in practical systems, due to the limited DoF in channel phase optimization for $L < M$ delay taps, the channel gains over multiple subcarriers cannot be exactly identical. Consequently, a performance gap arises between the MA-OFDM achievable rate and its upper bound given in Theorem \ref{Theo_rate}. Nevertheless, we will demonstrate in Section V by simulation that this performance gap can be small when the size of regions for antenna movement is sufficiently large.

\section{Optimization Algorithms}
In this section, we develop the PGA algorithm to numerically solve problem \eqref{eq_problem} subject to finite-size Tx/Rx regions. Then, a simplified PGA algorithm is also provided to reduce the computational complexity.
\subsection{PGA Algorithm}
For any given positions of the Tx-MA and Rx-MA, $\mathbf{t}$ and $\mathbf{r}$, problem \eqref{eq_problem} is convex w.r.t. $\mathbf{p}$ and its optimal (water-filling)  solution is given by
\begin{equation}\label{eq_power_allo_tr}
	p_{m}^{\star}(\mathbf{t}, \mathbf{r}) = \max \left\{\mu(\mathbf{t}, \mathbf{r}) - \frac{\sigma^{2}}{|c_{m}(\mathbf{t},\mathbf{r})|^{2}}, 0\right\},~1 \leq m \leq M,
\end{equation}
where $\mu(\mathbf{t}, \mathbf{r})$ should be selected to guarantee $\sum_{m=1}^{M} p_{m}^{\star}(\mathbf{t}, \mathbf{r}) = P$ and it can be calculated by utilizing the bisection search given a required accuracy, $\epsilon_{p}$. Substituting \eqref{eq_power_allo_tr} into \eqref{eq_rate}, we can simplify the achievable rate as
\begin{equation}\label{eq_rate_power}
	R_{p}(\mathbf{t},\mathbf{r}) = \frac{1}{M+M_{\mathrm{CP}}} \sum \limits_{m=1}^{M} \log_{2} \left(1+\frac{|c_{m}(\mathbf{t},\mathbf{r})|^{2}p_{m}^{\star}(\mathbf{t}, \mathbf{r})}{\sigma^{2}}\right).
\end{equation}

Then, problem \eqref{eq_problem} is equivalently transformed into
\begin{equation}\label{eq_problem2}
		\mathop{\max}\limits_{\mathbf{t} \in \mathcal{C}_{\mathrm{t}}, \mathbf{r} \in \mathcal{C}_{\mathrm{r}}}~R_{p}(\mathbf{t},\mathbf{r}).
\end{equation}
Since $R_{p}(\mathbf{t},\mathbf{r})$ is highly non-linear w.r.t. $\mathbf{t}$ and $\mathbf{r}$, there may exist a large number of local maxima in their feasible region. Conventional optimization methods may thus be trapped in locally optimal solutions of problem \eqref{eq_problem2}. To address this issue, we propose the PGA algorithm to efficiently search multiple maximum points of $(\mathbf{t}, \mathbf{r})$ and select the one yielding the highest achievable rate.

Specifically, we define $K_{\max} (\geq 1)$ as the maximum number of points for parallel search. We initialize $K^{(0)} = K_{\max}$ MA positioning solutions for the Tx-MA and Rx-MA as the candidate set $\bar{\mathcal{S}}_{0}=\{(\mathbf{t}_{k}^{(0)}, \mathbf{r}_{k}^{(0)})\}_{1 \leq k \leq K^{(0)}}$. Then, for the $i$-th iteration, $1 \leq i \leq I_{\max}$, we set the searching line segment for the $k$-th MA positioning vector, $1 \leq k \leq K^{(i-1)}$, as
\begin{equation}\label{eq_line_search}
	\left[\begin{aligned}
		&\mathbf{t}\\
		&\mathbf{r}
	\end{aligned}\right]
	=\left[\begin{aligned}
		&\mathbf{t}_{k}^{(i-1)}\\
		&\mathbf{r}_{k}^{(i-1)}
	\end{aligned}\right]
	+ \eta_{k} \times \left[\begin{aligned}
		&\nabla_{\mathbf{t}} R_{p}(\mathbf{t}_{k}^{(i-1)},\mathbf{r}_{k}^{(i-1)})\\
		&\nabla_{\mathbf{r}} R_{p}(\mathbf{t}_{k}^{(i-1)},\mathbf{r}_{k}^{(i-1)})
	\end{aligned}\right],
\end{equation}
where $0 < \eta_{k} \leq \eta_{k}^{\max}$ guarantees the search over the gradient ascent direction and $\eta_{k}^{\max}$ denotes the maximum value of $\eta_{k}$ which yields the MA positioning vector on the boundary of its feasible region. Note that two special cases should be considered for the searching line segment in \eqref{eq_line_search}. On one hand, if the gradient in \eqref{eq_line_search} is a zero vector, $R_{p}(\mathbf{t}_{k}^{(i-1)}, \mathbf{r}_{k}^{(i-1)})$ yields a local maximum w.r.t. $\mathbf{t}$ and $\mathbf{r}$. On the other hand, $\eta_{k}^{\max}=0$ indicates that $(\mathbf{t}_{k}^{(i-1)}, \mathbf{r}_{k}^{(i-1)})$ is located at the boundary of its feasible region. For both cases, the searching line segment in \eqref{eq_line_search} is degraded into an empty set and should not be used for line search. 

Due to the highly non-linearity of $R_{p}(\mathbf{t},\mathbf{r})$ w.r.t. $\mathbf{t}$ and $\mathbf{r}$, multiple local maximum points w.r.t. $\eta_{k}$ may exist over the searching line segment in \eqref{eq_line_search}. To calculate the local maximum points w.r.t. $\eta_{k}$ for each searching line segment, we may gradually increase $\eta_{k}$ from $0$ by a small positive step size $\zeta$, i.e., $\eta_{k}^{q} = q\zeta$, $q \geq 1, q \in \mathbb{Z}$, until $\eta_{k}^{q}$ achieves its maximum value $\eta_{k}^{\max}$. Denote the achievable rate function w.r.t. $\eta_{k}$ over the $k$-th line segment as
\begin{equation}\label{eq_rate_eta}
	R_{i-1}^{(k)}(\eta_{k}) \triangleq R_{p} \left( \tilde{\mathbf{t}}_{k}^{(i-1)}(\eta_{k}), \tilde{\mathbf{r}}_{k}^{(i-1)}(\eta_{k}) \right),
\end{equation}
with $\tilde{\mathbf{t}}_{k}^{(i-1)}(\eta_{k}) \triangleq \mathbf{t}_{k}^{(i-1)} + \eta_{k} \nabla_{\mathbf{t}} R_{p}(\mathbf{t}_{k}^{(i-1)},\mathbf{r}_{k}^{(i-1)})$ and $\tilde{\mathbf{r}}_{k}^{(i-1)}(\eta_{k}) \triangleq \mathbf{r}_{k}^{(i-1)}+ \eta_{k} \nabla_{\mathbf{r}} R_{p}(\mathbf{t}_{k}^{(i-1)},\mathbf{r}_{k}^{(i-1)})$. The local maximum points of $\eta_{k}$ are thus defined as the ones satisfying $R_{i-1}^{(k)}(\eta_{k}^{q}) > R_{i-1}^{(k)}(\eta_{k}^{q-1})$ and $R_{i-1}^{(k)}(\eta_{k}^{q}) \geq R_{i-1}^{(k)}(\eta_{k}^{q+1})$.

We collect all local maximum points w.r.t. $\eta_{k}$ over all $K^{(i-1)}$ searching line segments and denote the set of all corresponding MA positioning vectors as $\mathcal{S}_{i}$, with $N_{i}=|\mathcal{S}_{i}|$. To balance the computational complexity and achievable-rate performance, we employ a greedy scheme by selecting the largest $K^{(i)} = \min\{N_{i}, K_{\max}\}$ maxima from $\mathcal{S}_{i}$ for the subsequent iteration, where the shortlisted set of candidate MA positioning vectors are denoted by $\bar{\mathcal{S}}_{i}=\{(\mathbf{t}_{k}^{(i)}, \mathbf{r}_{k}^{(i)})\}_{1 \leq k \leq K^{(i)}}$. 
For the $i$-th iteration, the optimal solution for the MA positioning vector is updated by 
\begin{equation}\label{eq_rate_power_sel}
	(\mathbf{t}^{(i)}, \mathbf{r}^{(i)})=\mathop{\arg~\max}\limits_{(\mathbf{t},\mathbf{r}) \in \bar{\mathcal{S}}}~ R_{p}(\mathbf{t},\mathbf{r}),
\end{equation}
where $\bar{\mathcal{S}} = \bigcup_{j=0}^{i} \bar{\mathcal{S}}_{j}$ is the union set of all candidate MA positioning vectors during the iterations.

\begin{algorithm}[t]
	\caption{PGA algorithm for solving problem \eqref{eq_problem}.}
	\label{alg_PGA}
	\begin{algorithmic}[1]
		\REQUIRE ~$M$, $M_{\mathrm{CP}}$, $P$, $\sigma^{2}$, $\sigma^{2}$, $\lambda$, $\mathcal{C}_{\mathrm{t}}$, $\mathcal{C}_{\mathrm{r}}$, $T$, $\{L_{\tau}\}$, $\{\mathbf{\Sigma}_{\tau}\}$, $\{\hat{\theta}_{\tau}^{\ell}\}$, $\{\hat{\phi}_{\tau}^{\ell}\}$, $\{\tilde{\theta}_{\tau}^{\ell}\}$, $\{\tilde{\phi}_{\tau}^{\ell}\}$, $K_{\max}$, $\zeta$, $I_{\max}$, $\epsilon_{p}$.
		\ENSURE ~$\mathbf{t}^{\star}$, $\mathbf{r}^{\star}$, $\mathbf{p}^{\star}$. \\
		\STATE Initialize an empty set of local maximum points, $\bar{\mathcal{S}}$.
		\STATE Initialize $K^{(0)} = K_{\max}$.
		\STATE Initialize candidate set $\bar{\mathcal{S}}_{0}=\{(\mathbf{t}_{k}^{(0)}, \mathbf{r}_{k}^{(0)})\}_{1 \leq k \leq K^{(0)}}$.
		\STATE Update $\bar{\mathcal{S}} \leftarrow \bar{\mathcal{S}} \cup \bar{\mathcal{S}}_{0}$.
		\STATE Update $(\mathbf{t}^{(0)}, \mathbf{r}^{(0)})$ according to \eqref{eq_rate_power_sel}.
		\FOR   {$i=1:1:I_{\max}$}
			\STATE Initialize an empty set $\mathcal{S}_{i}$.
			\FOR   {$k=1:1:K^{(i-1)}$}
				\STATE Calculate gradient $\nabla_{\mathbf{t}} R_{p}(\mathbf{t}_{k}^{(i-1)},\mathbf{r}_{k}^{(i-1)})$ and $\nabla_{\mathbf{r}} R_{p}(\mathbf{t}_{k}^{(i-1)},\mathbf{r}_{k}^{(i-1)})$ in \eqref{eq_line_search}.
				\STATE Calculate achievable rate w.r.t. $\eta_{k}$ according to \eqref{eq_rate_eta}.
				\IF {$\forall \eta_{k} \in (0, \eta_{k}^{\max}]$ is a local maximum point}
					\STATE Update $\mathcal{S}_{i} \leftarrow \mathcal{S}_{i} \cup \left\{(\tilde{\mathbf{t}}_{k}^{(i-1)}(\eta_{k}), \tilde{\mathbf{r}}_{k}^{(i-1)}(\eta_{k})) \right\}$. 
				\ENDIF
			\ENDFOR 
			\IF {$\mathcal{S}_{i}$ is empty}
				\STATE Set $(\mathbf{t}^{(i)}, \mathbf{r}^{(i)}) = (\mathbf{t}^{(i-1)}, \mathbf{r}^{(i-1)})$ .
				\STATE Break.
			\ENDIF
			\STATE Update $K^{(i)} = \min\{N_{i}, K_{\max}\}$ with $N_{i}=|\mathcal{S}_{i}|$.
			\STATE Update candidate set $\bar{\mathcal{S}}_{i}=\{(\mathbf{t}_{k}^{(i)}, \mathbf{r}_{k}^{(i)})\}_{1 \leq k \leq K^{(i)}} \subseteq \mathcal{S}_{i}$ yielding the $K^{(i)}$ highest achievable rate.
			\STATE Update $\bar{\mathcal{S}} \leftarrow \bar{\mathcal{S}} \cup \bar{\mathcal{S}}_{i}$.
			\STATE Update $(\mathbf{t}^{(i)}, \mathbf{r}^{(i)})$ according to \eqref{eq_rate_power_sel}.
		\ENDFOR
		\STATE Set MA positioning vectors as $\mathbf{t}^{\star}=\mathbf{t}^{(i)}$ and $\mathbf{r}^{\star}=\mathbf{r}^{(i)}$.
		\STATE Calculate the corresponding $\mathbf{p}^{\star}$ according to \eqref{eq_power_allo_tr}.
		\RETURN $\mathbf{t}^{\star}$, $\mathbf{r}^{\star}$, $\mathbf{p}^{\star}$.
	\end{algorithmic}
\end{algorithm}

The proposed PGA algorithm for solving problem \eqref{eq_problem} is summarized in Algorithm \ref{alg_PGA}. Specifically, for each iteration, we calculate the local maximum points w.r.t. $\eta_{k}$ of all searching segments and collect the corresponding MA positioning vectors in lines 8-14. The shortlisted local maximum points w.r.t. $\eta_{k}$ and the optimal MA positioning vector for each iteration are updated in lines 19-22. The iterations are repeated until the maximum iteration index, $I_{\max}$, is achieved or $\mathcal{S}_{i}$ is an empty set. Since the set of all candidate MA positioning vectors is accumulatively enlarged during the iterations shown in line 21, the achievable rate $R_{p}(\mathbf{t}^{(i)}, \mathbf{r}^{(i)})$ is non-decreasing with the iteration index. Given the fact that $R_{p}(\mathbf{t},\mathbf{r})$ is bounded from its definition, the convergence of Algorithm \ref{alg_PGA} is guaranteed. Note that if the maximum number of candidate MA positioning vector, $K_{\max}$, is set to a large value, the proposed PGA algorithm can find more local maxima for problem \eqref{eq_problem}, which can approach the globally optimal solution with a high possibility. However, large $K_{\max}$ entails high computational complexity because the space for line search increases for each iteration. Thus, a trade-off exists between the achievable-rate performance and the computational complexity by choosing different values of $K_{\max}$. In fact, for the special case of $K_{\max}=1$, we have $K^{(i)} = \min\{N_{i}, K_{\max}\}=1$ during the iterations, and thus the proposed PGA algorithm degrades into the classical steepest gradient ascent method \cite{boyd2004convex}.

The computational complexity of Algorithm \ref{alg_PGA} is analyzed as follows. For any given MA positioning vector, the computational complexity for calculating the CFR vector $\mathbf{c}(\mathbf{t},\mathbf{r})$ in \eqref{eq_CFR} is $\mathcal{O}(M)$. The complexity for calculating power allocation $p_{m}^{\star}(\mathbf{t}, \mathbf{r})$ in \eqref{eq_power_allo_tr} based on bisection search is $\mathcal{O}(\log_{2} \frac{1}{\epsilon_{p}})$, where $\epsilon_{p}$ is the search accuracy. For each iteration, the total number of points over all searching line segments is no larger than $K_{\max} A/\zeta$, where $A= \max\{\sqrt[3]{|\mathcal{C}_{\mathrm{t}}|}, \sqrt[3]{|\mathcal{C}_{\mathrm{r}}|}\}$ is the maximum size of the Tx/Rx region and $\zeta$ is the step size for line search. Thus, the total computational complexity of Algorithm \ref{alg_PGA} is given by $\mathcal{O}\left(\frac{I_{\max} K_{\max} A}{\zeta} \left(M + \log_{2} \frac{1}{\epsilon_{p}}\right) \right)$.

\subsection{Simplified PGA Algorithm}
The proposed Algorithm \ref{alg_PGA} requires to calculate the optimal power allocation for any given MA positioning vector by using the bisection search method, which entails a high computational complexity for line search over all candidate segments in \eqref{eq_line_search}. To further reduce the computational complexity, we propose a simplified PGA solution for problem \eqref{eq_problem} in this subsection. 

Recall the proof of Theorem \ref{Theo_rate}, where the average channel power gain over all $M$ subcarriers is equal to the CIR vector's power, $\left\|\mathbf{h}(\mathbf{t},\mathbf{r})\right\|_{2}^{2}$. Thus, the maximization of $\left\|\mathbf{h}(\mathbf{t},\mathbf{r})\right\|_{2}^{2}$ can help increase the achievable rate of the considered MA-OFDM system, which simplifies the calculation of power allocation over the iterations. Thus, problem \eqref{eq_problem2} can be simplified as
\begin{equation}\label{eq_problem3}
	\mathop{\max}\limits_{\mathbf{t} \in \mathcal{C}_{\mathrm{t}}, \mathbf{r} \in \mathcal{C}_{\mathrm{r}}}~\left\|\mathbf{h}(\mathbf{t},\mathbf{r})\right\|_{2}^{2}.
\end{equation}
The proposed PGA algorithm can be also applied to solve problem \eqref{eq_problem3}, whereas the calculations of the gradient in line 9 and the objective function in line 10 can be significantly simplified without involving the power allocation. In particular, the gradient of $\left\|\mathbf{h}(\mathbf{t},\mathbf{r})\right\|_{2}^{2}$ w.r.t.  $(\mathbf{t},\mathbf{r})$ can be derived in closed form as follows:
\begin{equation}\label{eq_gradient}
	\begin{aligned}
		&\nabla_{\mathbf{t}} \left\|\mathbf{h}(\mathbf{t},\mathbf{r})\right\|_{2}^{2} = \sum \limits_{\tau=1}^{T} \sum \limits_{\ell=1}^{L_{\tau}} \sum \limits_{\tau'=1}^{T} \sum \limits_{\ell'=1}^{L_{\tau'}} -\frac{2\pi}{\lambda}|b_{\tau}^{\ell}| |b_{\tau'}^{\ell'}| (\hat{\mathbf{k}}_{\tau}^{\ell}-\hat{\mathbf{k}}_{\tau'}^{\ell'}) \times \\
		& \sin\left[\frac{2\pi}{\lambda} \mathbf{t}^{\mathrm{T}} (\hat{\mathbf{k}}_{\tau}^{\ell}-\hat{\mathbf{k}}_{\tau'}^{\ell'}) - \frac{2\pi}{\lambda} \mathbf{r}^{\mathrm{T}} (\tilde{\mathbf{k}}_{\tau}^{\ell}-\tilde{\mathbf{k}}_{\tau'}^{\ell'})+(\angle b_{\tau}^{\ell} - \angle b_{\tau'}^{\ell'})\right],\\
		&\nabla_{\mathbf{r}} \left\|\mathbf{h}(\mathbf{t},\mathbf{r})\right\|_{2}^{2} = \sum \limits_{\tau=1}^{T} \sum \limits_{\ell=1}^{L_{\tau}} \sum \limits_{\tau'=1}^{T} \sum \limits_{\ell'=1}^{L_{\tau'}} \frac{2\pi}{\lambda}|b_{\tau}^{\ell}| |b_{\tau'}^{\ell'}| (\tilde{\mathbf{k}}_{\tau}^{\ell}-\tilde{\mathbf{k}}_{\tau'}^{\ell'}) \times \\
		& \sin\left[\frac{2\pi}{\lambda} \mathbf{t}^{\mathrm{T}} (\hat{\mathbf{k}}_{\tau}^{\ell}-\hat{\mathbf{k}}_{\tau'}^{\ell'}) - \frac{2\pi}{\lambda} \mathbf{r}^{\mathrm{T}} (\tilde{\mathbf{k}}_{\tau}^{\ell}-\tilde{\mathbf{k}}_{\tau'}^{\ell'})+(\angle b_{\tau}^{\ell} - \angle b_{\tau'}^{\ell'})\right].
	\end{aligned}
\end{equation}
After obtaining the solution for MA positioning vectors, the power allocation is finally calculated according to \eqref{eq_power_allo_tr}. As such, the computational complexity of the simplified PGA algorithm is given by $\mathcal{O}\left(\frac{I_{\max} K_{\max} A}{\zeta} + M + \log_{2} \frac{1}{\epsilon_{p}} \right)$, which is much lower than that of the PGA algorithm.

\section{Simulation Results}
In this section, we present simulation results to verify the analytical results and optimization algorithms for the proposed MA-OFDM wideband communication system. The simulation setup is first illustrated and then the numerical results are presented.

\subsection{Simulation Setup and Benchmark Schemes}
In the simulation, the carrier frequency and bandwidth are set as $f_{c}=2.4$ GHz and $B=40$ MHz, respectively. The number of OFDM subcarriers and the CP length are set as $M=64$ and $M_{\mathrm{CP}}=6$, respectively. The maximum transmit power is set as $P=1$ Watt (W), i.e., 30 dBm. The noise power of each subcarrier is given by $\sigma^{2}=\frac{1}{M}10^{\frac{BN_{0}}{10}-3}$ W, where $N_{0}=-174$ dBm/Hz represents the power spectral density of noise. The Tx and Rx regions for antenna moving are both set as cubes of size $A=4\lambda$, i.e., $\mathcal{C}_{\mathrm{t}}=\mathcal{C}_{\mathrm{r}}=[-2\lambda, 2\lambda] \times [-2\lambda, 2\lambda] \times [-2\lambda, 2\lambda]$. The number of non-zero delay taps and the number of channel paths per tap are set as $T=6$ and $L=5$, respectively \cite{zheng2020IRSOFDM}. The power delay profile follows the exponential decay \cite {riihonen2010generalized} with $q_{\tau}=\frac{1}{\xi} e^{-\alpha(\tau-1)}$, $1 \leq \tau \leq T$, where $\xi = \sum_{\tau=1}^{T} e^{-\alpha(\tau-1)}$ is the normalization factor and $\alpha$ is the exponential decay factor set as 2. For each tap, the coefficients of multiple channel paths are modeled as independent and identically distributed (i.i.d.) complex Gaussian random variables following $b_{\tau}^{\ell} \sim \mathcal{CN}(0, g_{0}q_{\tau}/L)$, $1 \leq \ell \leq L$, $1 \leq \tau \leq T$, where $g_{0}$ denotes the large-scale average channel power gain. As such, the average SNR at the Rx is given by $\mathrm{SNR}=\frac{g_{0}P}{M\sigma^{2}}$, which is set as 25 dB. For each channel path, the AoDs and AoAs follow the joint uniform distribution with the PDF given by $f_{\mathrm{AoD/AoA}}(\theta, \phi)= \frac{\cos \theta}{2\pi}$ \cite{zhu2022MAmodel}. For Algorithm \ref{alg_PGA}, the maximum number of iterations is set as $I_{\max}=100$. The maximum number of candidate solutions for parallel search is set as $K_{\max}=10$. The accuracy for power allocation in \eqref{eq_power_allo_tr} is set as $\epsilon_{p}=10^{-6}$. In the simulation figures, each point is averaged over $10^4$ random channel realizations.

The upper bound on the achievable rate is defined in \eqref{eq_rate_bound}. For the FPA benchmark scheme, the Tx and Rx antennas are fixed at the reference point of their corresponding regions. For the antenna selection (AS) benchmark scheme, three FPAs are deployed at both the Tx and Rx with half-wavelength spacing. For each channel realization, the Tx-Rx antenna pair which achieves the maximum achievable rate is selected via exhaustive search. For both FPA and AS schemes, the water-filling power allocation in \eqref{eq_power_allo_tr} is utilized for maximizing the achievable rate.

\subsection{Numerical Results}
\begin{figure}[t]
	\centering
	\includegraphics[width=\figwidth cm]{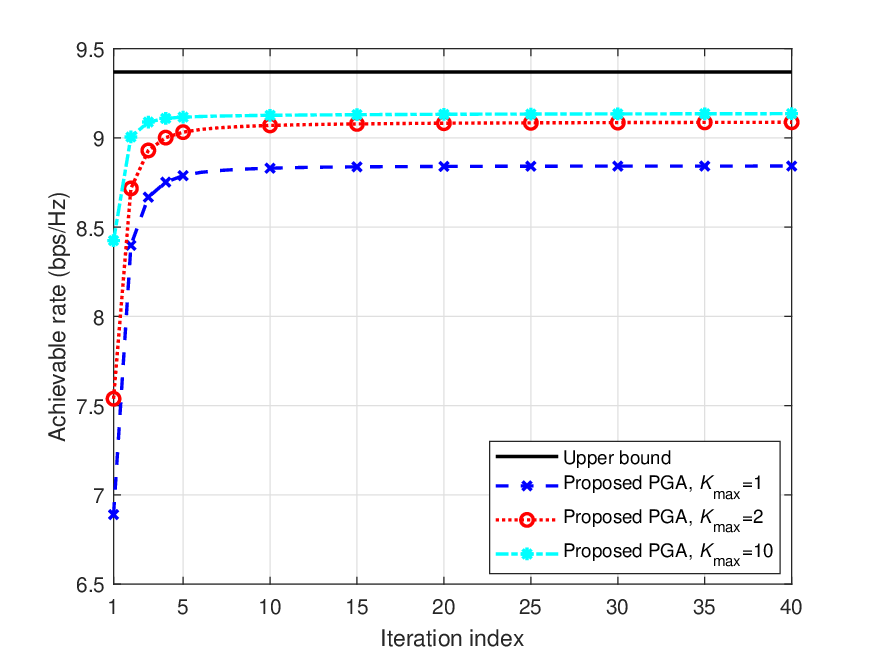}
	\caption{Evaluation of the convergence of the proposed PGA algorithm under different $K_{\max}$.}
	\label{fig:iteration_rate}
\end{figure}
\begin{figure}[t]
	\centering
	\includegraphics[width=\figwidth cm]{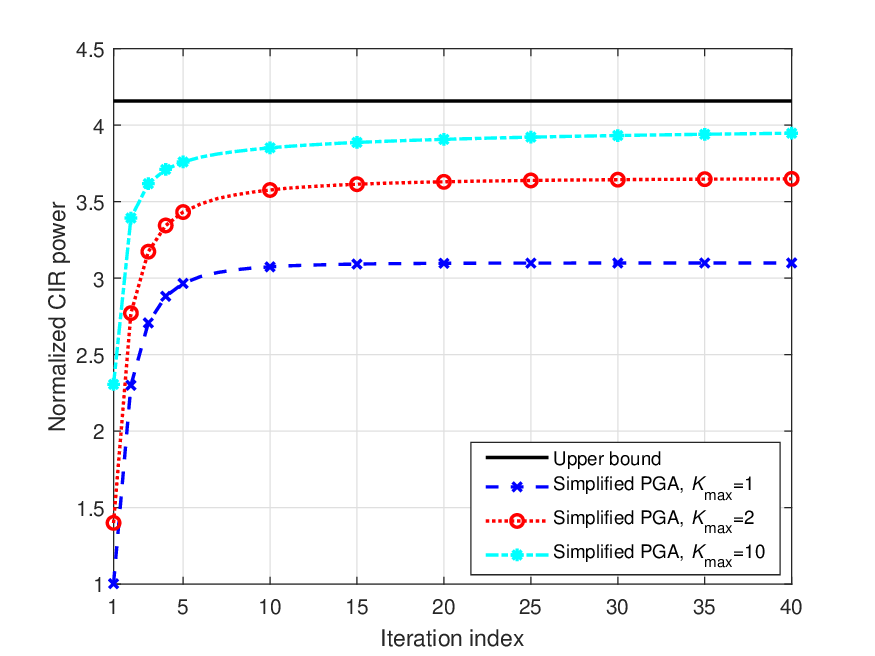}
	\caption{Evaluation of the convergence of the proposed simplified PGA algorithm under different $K_{\max}$.}
	\label{fig:iteration_CIR}
\end{figure}

First, we evaluate the convergence performance of the proposed PGA algorithm in Fig. \ref{fig:iteration_rate}. As can be observed, the achievable rate increases rapidly with the iteration index and converges after 10 iterations, which validates the convergence of Algorithm \ref{alg_PGA}. Besides, for the proposed PGA algorithm, the maximum number  of candidate solutions, $K_{\max}$, for parallel search significantly impacts the performance of the solution. For example, if $K_{\max}=1$, the proposed PGA algorithm degrades into the steepest ascent method and thus it becomes highly likely to obtain a locally optimal solution for problem \eqref{eq_problem}. In comparison, as $K_{\max}$ increases, the proposed PGA algorithm can simultaneously search over multiple line segments during the iterations and thus can obtain the globally optimal solution with a higher possibility. It is worth noting that for sufficiently large $K_{\max}$ (e.g., 10), the proposed solution closely approaches the performance upper bound on the OFDM achievable rate although it is derived under the assumption of infinite size of the Tx/Rx region. The results in Fig. \ref{fig:iteration_rate} demonstrate the efficacy of the proposed PGA algorithm for obtaining a near-optimal MA positioning vector in the given Tx/Rx region with a finite size. 

Next, we evaluate the convergence performance of the simplified PGA algorithm in Fig. \ref{fig:iteration_CIR}. In particular, the power of CIR in \eqref{eq_problem3} is normalized by the large-scale channel power gain, i.e., $\left\|\mathbf{h}(\mathbf{t},\mathbf{r})\right\|_{2}^{2} / g_{0}$. The corresponding upper bound on the normalized CIR power is thus given by $\sum_{\tau=1}^{T} \|\mathbf{b}_{\tau}\|_{1}^{2}/g_{0}$ in \eqref{eq_CIR_bound}. It can be observed from Fig. \ref{fig:iteration_CIR} that the simplified PGA algorithm for maximizing the CIR power converges after 20 iterations. The achieved CIR power increases with the maximum number of candidate solutions, $K_{\max}$, for parallel search. Moreover, for sufficiently large $K_{\max}$, the CIR power under the obtained MA positioning vector closely approaches the performance upper bound.

\begin{figure}[t]
	\centering
	\subfigure[Achievable rate in bps/Hz]{\includegraphics[width=\figwidth cm]{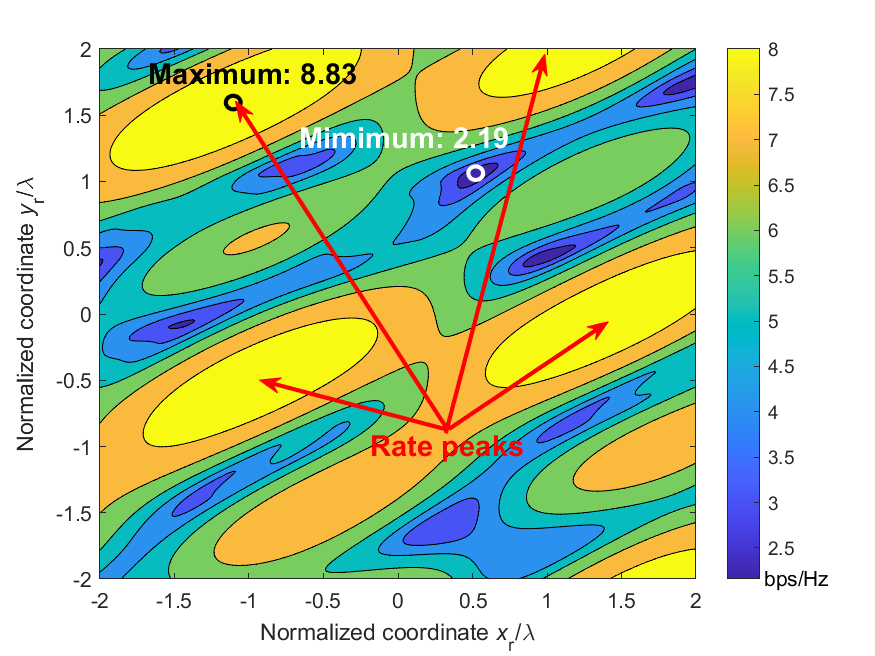} \label{fig:demo_rate}}
	\subfigure[Normalized CIR power]{\includegraphics[width=\figwidth cm]{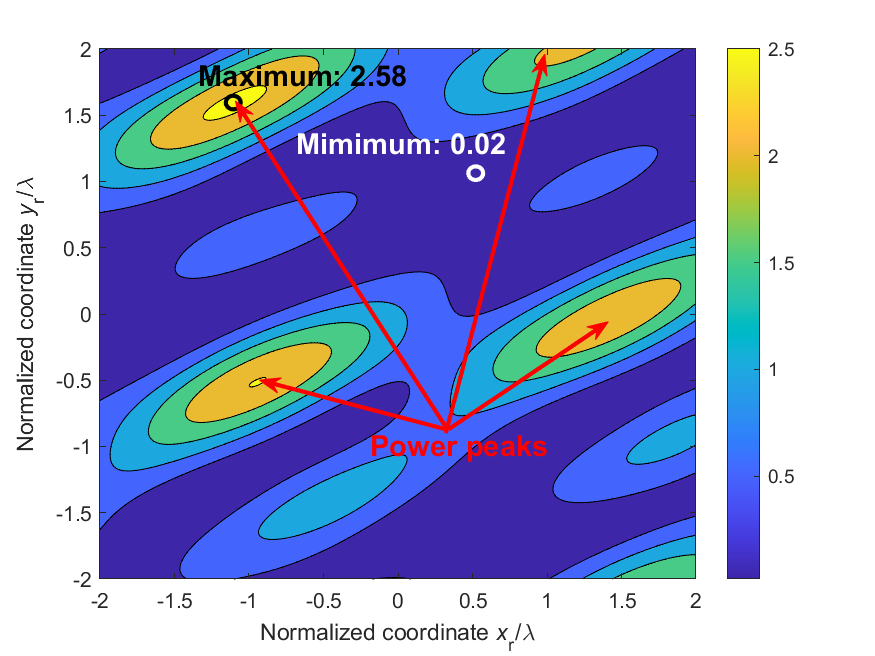} \label{fig:demo_channel}}
	\caption{Demonstration of the variation of the achievable rate and the normalized CIR power within the Rx region.}
	\label{fig:demo}
\end{figure}

To shed more light on the impact of MA position optimization, we demonstrate in Fig. \ref{fig:demo} the variation of the achievable rate and the normalized CIR power within the Rx region, where the Tx antenna is fixed at the origin of its local coordinate system, i.e., $\mathbf{t}=[0,0,0]^{\mathrm{T}}$ and $\mathbf{r}=[x_{\mathrm{r}},y_{\mathrm{r}},0]^{\mathrm{T}}$. As can be observed, the maps of achievable rate and CIR power exhibit high correlation within the Rx region. There are four peaks of the achievable rate in Fig. \ref{fig:demo_rate}, which are achieved at approximately the same positions with those of the CIR power in Fig. \ref{fig:demo_channel}. This alignment validates that the simplified PGA algorithm based on CIR power maximization can effectively maximize the achievable rate of the considered MA-OFDM communication system. Moreover, the achievable rate can increase from the minimum value 2.19 bps/Hz to the maximum value 8.83 bps/Hz between two locations with a distance of no larger than $2\lambda$. Similarly, the normalized CIR power can increase from 0.02 to 2.58, which yields over 20 dB gain in the channel power. The results in Fig. \ref{fig:demo} confirm the efficacy of the MA-OFDM communication system via antenna position optimization, even within small regions.

\begin{figure}[t]
	\centering
	\includegraphics[width=\figwidth cm]{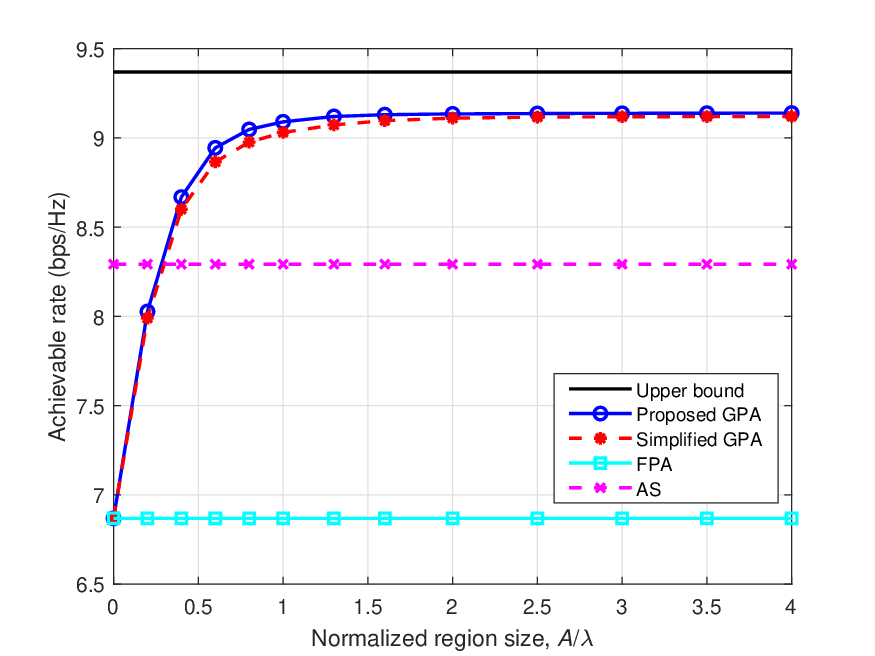}
	\caption{Achievable rates of the proposed and benchmark schemes versus the normalized region sizes for moving antennas at the Tx/Rx.}
	\label{fig:region}
\end{figure}

In Fig. \ref{fig:region}, we show the achievable rates of the proposed solutions for MA-OFDM systems versus the region sizes for moving antennas at the Tx/Rx and compare them with benchmark schemes. Both the proposed PGA and simplified PGA algorithms can reap an increasing achievable rate as the region size becomes larger. This is due to more DoFs in optimizing the antennas' positions with larger Tx/Rx region sizes. Moreover, the performance gap between the proposed and simplified PGA algorithms is small, especially for a large region size. It demonstrates the effectiveness of the simplified algorithm by maximizing the CIR power. In addition, the achievable rates of the proposed MA-OFDM systems are higher than those of both the conventional FPA and AS systems and can approach the performance upper bound when the size of the Tx/Rx region is larger than $\lambda$. This indicates that the achievable rate performance of MA-OFDM systems can be significantly improved by moving the Tx-MA/Rx-MA within a small local region with the size in the order of wavelength.

\begin{figure}[t]
	\centering
	\includegraphics[width=\figwidth cm]{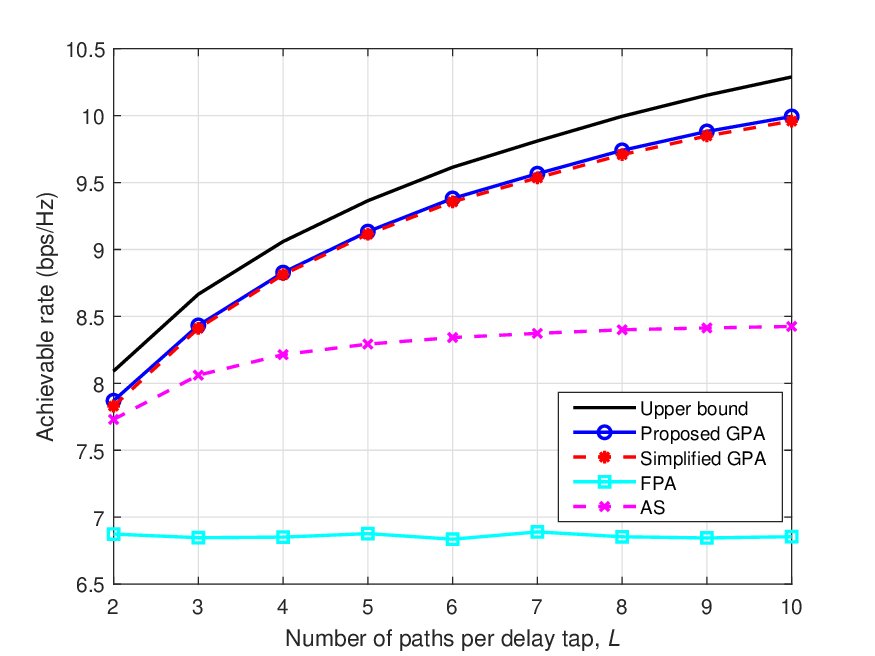}
	\caption{Achievable rates of the proposed and benchmark schemes versus the number of channel paths per delay tap.}
	\label{fig:path}
\end{figure}

Fig. \ref{fig:path} shows the achievable rates of different schemes versus the number of channel paths per delay tap, $L$. As $L$ increases, the small-scale fading of the channel in the spatial domain becomes more pronounced for each clustered tap, and thus the MAs' position optimization can reap higher performance gains in maximizing the achievable rate. In comparison, the FPA and AS schemes cannot exploit such spatial diversity, with their achievable rates non-increasing or slowly increasing with $L$. For example, the proposed MA-OFDM system for $L=3$ can achieve an achievable-rate boost of 1.6 bps/Hz and 0.5 bps/Hz compared to FPA and AS systems, respectively. The corresponding rate boost increases to 3.1 bps/Hz and 1.5 bps/Hz for $L=10$, respectively. Moreover, the performance difference between the proposed and simplified PGA algorithms is negligible and their performance gap to the upper bound is no larger than 0.25 bps/Hz in terms of achievable rate.

\begin{figure}[t]
	\centering
	\includegraphics[width=\figwidth cm]{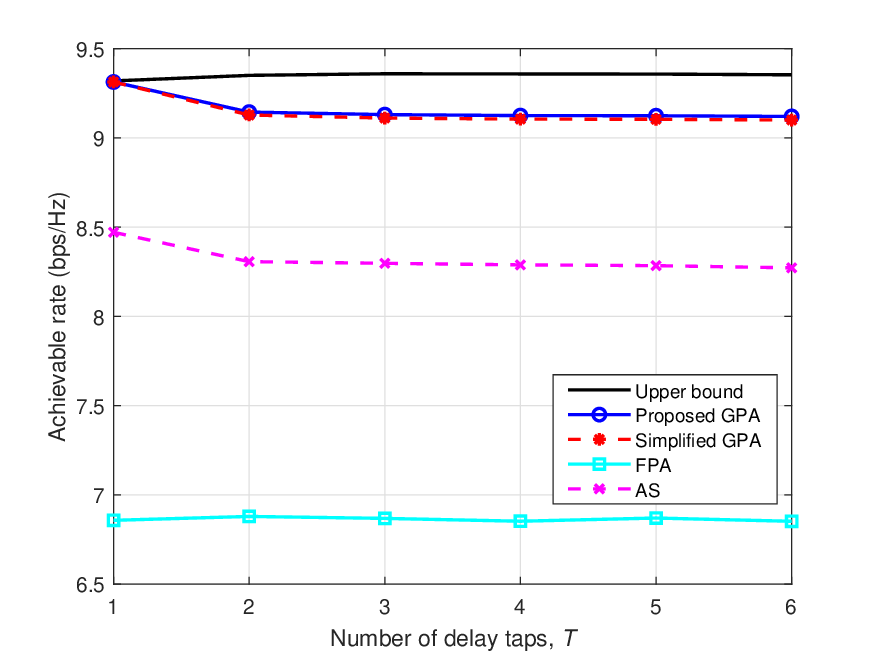}
	\caption{Achievable rates of the proposed and benchmark schemes versus the number of delay taps.}
	\label{fig:tap}
\end{figure}

In Fig. \ref{fig:tap}, we illustrate the achievable rates of different schemes versus the number of delay taps, $T$. It is observed that the achievable rates of the MA and AS systems both decrease with $T$. This is because the increasing number of delay taps in the time domain results in a more pronounced frequency-selective fading of wireless channels. Thus, it becomes more challenging to realize equal channel gain realization over all OFDM subcarriers as shown in Theorem \ref{Theo_rate} via MAs' position optimization in the spatial domain. Nonetheless, the CIR power over each delay tap can also be increased by optimizing the positions of the Tx-MA and Rx-MA in larger spatial regions. From the results in Figs. \ref{fig:path} and \ref{fig:tap}, we infer that MAs exhibit significant superiority to conventional FPAs in scenarios where a large number of channel paths with angular diversity are present within a small delay spread. For example, when the Tx and Rx are both deployed in an indoor environment, the total number of channel paths (i.e., $T \times L$) is large due to the abundant scatterers. Meanwhile, since the scatterers are all located in the indoor area with limited space, the short distance of signal propagation renders a small delay spread of such channel paths (i.e., $T$).

\begin{figure}[t]
	\centering
	\includegraphics[width=\figwidth cm]{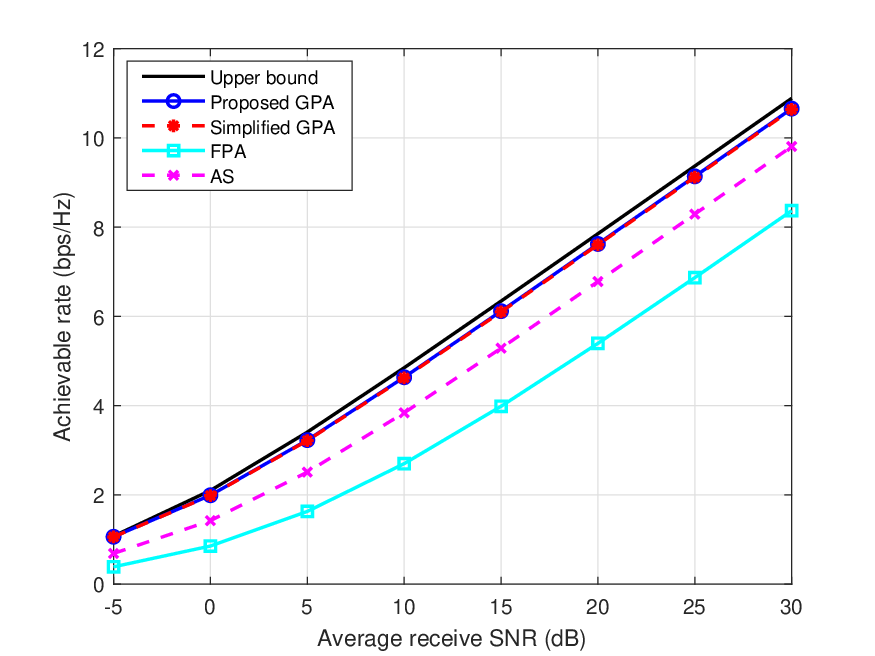}
	\caption{Achievable rates of the proposed and benchmark schemes versus the average SNR at the Rx.}
	\label{fig:SNR}
\end{figure}

Fig. \ref{fig:SNR} evaluates the achievable rates of the proposed and benchmark schemes versus the average SNR at the Rx. It is observed again that the PGA algorithms can achieve a performance close to the rate upper bound. To achieve the same achievable rate, the proposed MA system requires a lower average receive SNR compared to the FPA and AS systems, which can help reduce the transmit power under the same channel condition. As the average receive SNR increases to be larger than 15 dB, the performance gap between the MA and all benchmark schemes converges. Specifically, the MA system can reap 2.5 dB and 7 dB gain in decreasing the transmit power over the FPA and AS systems, respectively. Note that if the number of channel paths per delay tap and the region size for antenna moving increase, the performance gain of MAs over FPAs can be further improved. Moreover, the results in Fig. \ref{fig:SNR} are averaged over a large number of random channel realizations. For site-specific channel realizations, e.g., the channel in Fig. \ref{fig:demo}, the performance improvement provided by MA position optimization can be occasionally even larger, which significantly improves the worst-case performance of OFDM systems.

\begin{figure}[t]
	\centering
	\includegraphics[width=\figwidth cm]{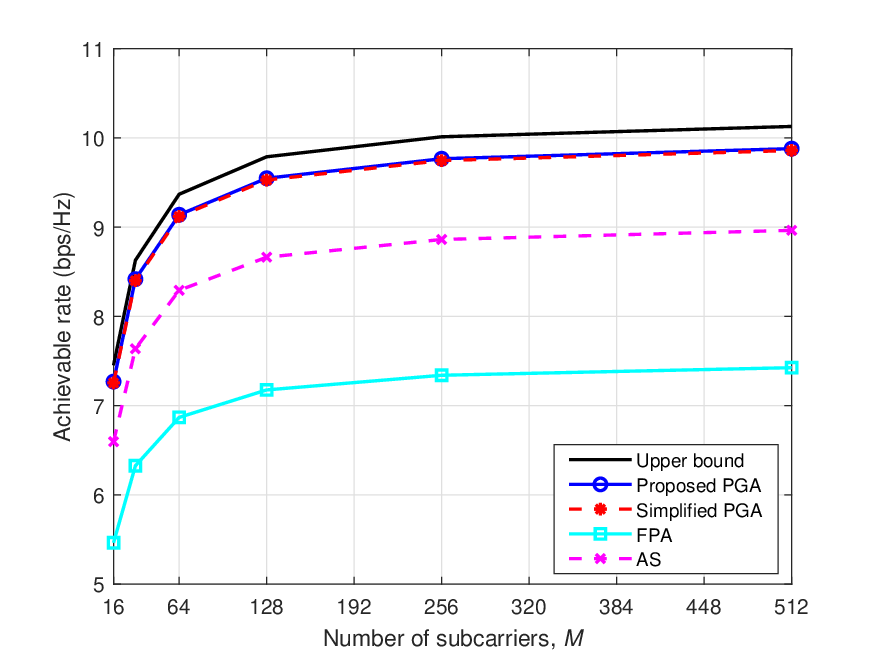}
	\caption{Achievable rates of the proposed and benchmark schemes versus the number of OFDM subcarriers.}
	\label{fig:subcarrier}
\end{figure}

Next, we show in Fig. \ref{fig:subcarrier} the achievable rates of different schemes with varying number of OFDM subcarriers. We can observe that the MA system always outperform FPA and AS systems under different $M$'s and the performance gap between our proposed algorithms and the upper bound does not exceed 0.3 bps/Hz in terms of achievable rate. The achievable rates of all schemes increase with $M$ because the ratio of the CP length to the OFDM symbol length decreases, which can help increase the spectral efficiency. However, the corresponding computational overhead also increases with $M$. For the considered system setup, the increment of the achievable rate is observed to be small if $M$ exceeds 128.

\begin{figure}[t]
	\centering
	\subfigure[$L=6$]{\includegraphics[width=\figwidth cm]{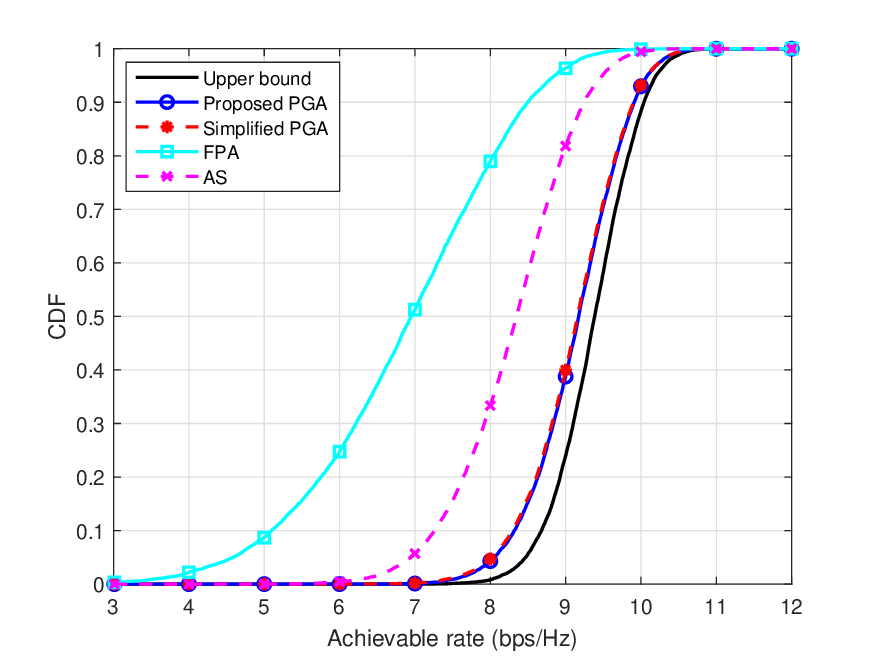} \label{fig:CDF_L6}}
	\subfigure[$L=10$]{\includegraphics[width=\figwidth cm]{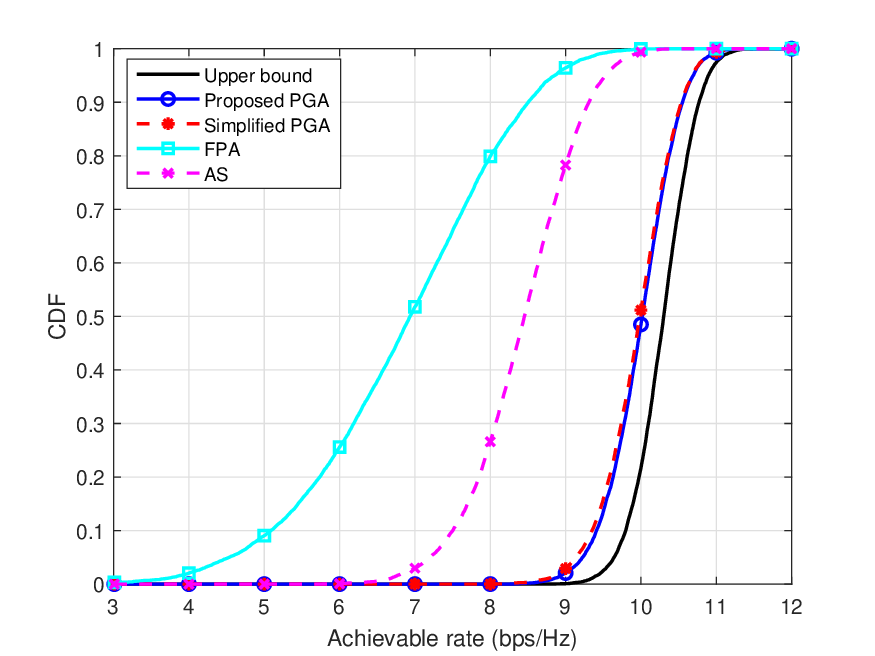} \label{fig:CDF_L10}}
	\caption{The experimental CDF of achievable rates for the proposed and benchmark schemes.}
	\label{fig:CDF}
\end{figure}

Finally, in Fig. \ref{fig:CDF}, we illustrate the experimental cumulative distribution function (CDF) of achievable rates for the proposed and benchmark schemes under different numbers of channel paths per delay tap, $L$. The experimental CDF is calculated based on $10^4$ random channel realizations, i.e., $\mathrm{CDF}(r) = |\mathcal{I}_{R \leq r}|/10^{4}$, where $|\mathcal{I}_{R \leq r}|$ denotes the number of channel realizations with the achievable rate no larger than $r$. It is observed that the CDF curves for the proposed and simplified PGA algorithms have small deviation for all achievable rates and can approach that for the upper bound. Besides, the proposed MA-OFDM system significantly outperforms both the FPA and AS systems in terms of outage probability, especially for a large number of channel paths per delay tap. For example, if we set the threshold of the achievable rate as $r=8$ bps/Hz, the outage probabilities of the MA, FPA, and AS systems for $L=6$ are 4.4\%, 79.5\%, and 33.4\%. respectively. For $L=10$, the outage probabilities of the MA, FPA, and AS systems are 0.03\%, 79.5\%, 26.6\%, respectively. Since the spatial variation of wireless channels is not exploited by the FPA system, it cannot reap the performance gain in terms of outage probability. In contrast, the MA and AS systems can leverage the spatial DoFs for increasing the achievable rate, and thus the outage performance is also improved. Compared to AS systems, the proposed MA system can yield higher performance gain in terms of outage probability because the MA can fully exploit the wireless channel variation in the continuous Tx/Rx regions, especially when $L$ is large.

\section{Conclusion}
In this paper, we investigated the MA-OFDM wideband communications under frequency-selective fading channels. A general multi-tap field-response channel model was adopted to characterize the wireless channel variations in both space and frequency w.r.t. different positions of the MAs at the Tx and Rx sides. We demonstrated the great potential of MA positioning for achieving the desired CIR with maximum channel gains yet arbitrary channel phases over all clustered delay taps. Based on this finding, an upper bound on the OFDM achievable rate was derived in closed form when the size of the Tx/Rx region for antenna movement is arbitrarily large. Furthermore, we developed a PGA algorithm to obtain locally optimal solutions to the MAs' positions for OFDM rate maximization subject to finite-size Tx/Rx regions. To  reduce computational complexity, a simplified PGA algorithm was also provided by maximizing the total channel power gain instead, where the transmit power allocation is not required over the iterations. Simulations were conducted to validate both analytical results and numerical solutions for MA-OFDM wideband communication systems. It was shown that the proposed PGA algorithm can approach the OFDM rate upper bound closely with the increase of Tx/Rx region sizes and outperforms conventional FPA-OFDM and AS-OFDM systems. Moreover, the simplified PGA algorithm with low computational complexities was shown to achieve a comparable achievable rate to the original PGA algorithm with a negligible performance gap. Simulation results also revealed that the proposed MA-OFDM system can yield more significant performance gains over its FPA counterpart under wideband wireless channels with a small number of delay taps each encompassing a large number of independent channel paths. Future research may consider the extension of MA-OFDM to multiuser/multi-antenna communication systems, by investigating their efficient channel estimation and MA position optimization schemes.

\appendices

\section{Proof of Lemma \ref{lemma_LIA}}\label{Appen_LIA}
For notation simplicity, we reorganize the index of the random vector as $\bm{\Theta}=[\Theta_{1}, \Theta_{2}, \cdots, \Theta_{N}]^{\mathrm{T}}$ and its joint PDF as $f_{\bm{\Theta}}(\bm{\vartheta})=\prod_{n=1}^{N} f_{\Theta_{n}}(\hat{\vartheta}_{n})$, with $N=\sum_{\tau=1}^{T} L_{\tau}$ and $\hat{\vartheta}_{n}$ denoting the $n$-th element of $\bm{\vartheta}$. We consider the Lebesgue measure $\mathcal{L}(\cdot)$ defined on real-number interval $[-1,1]$ \cite{bartle2014elements}. Let $\mathcal{Q}_{[-1,1]}$ denote the sets of rational numbers in interval $[-1,1]$. Since the set of rational numbers is countable, the Lebesgue measure of $\mathcal{Q}_{[-1,1]}$ is $\mathcal{L}(\mathcal{Q}_{[-1,1]}^{0}) = 0$. Denoting the upper bound on the PDFs $\{f_{\Theta_{n}}(\hat{\vartheta}_{n})\}$ as $\varrho$, the probability of $\Theta_{n}$ belonging to $\mathcal{Q}_{[-1,1]}$ is given by $\Pr\{\Theta_{n} \in \mathcal{Q}_{[-1,1]}\} \leq \varrho \times \mathcal{L}(\mathcal{Q}_{[-1,1]}) = 0$. Thus, we know that the elements in $\bm{\Theta}$ are irrational numbers with probability 1.

Next, we define $\mathcal{Q}_{[-1,1]}^{n}(\{\hat{\vartheta}_{j}\}_{1 \leq j \leq N}^{j \neq n}) = \mathcal{Q}_{[-1,1]} \oplus \{\hat{\vartheta}_{j}\}_{1 \leq j \leq N}^{j \neq n}$ as the set of numbers within interval $[-1,1]$ which are generated by the linear combinations of elements in $\mathcal{Q}_{[-1,1]} \cup \{\hat{\vartheta}_{j}\}_{1 \leq j \leq N}^{j \neq n}$ over $\mathbb{Q}$, i.e., $\mathcal{Q}_{[-1,1]}^{n}(\{\hat{\vartheta}_{j}\}_{1 \leq j \leq N}^{j \neq n}) = \{x \in [-1,1] \big{|} x = \sum_{q_{k} \in \mathcal{Q}_{[-1,1]} \cup \{\hat{\vartheta}_{j}\}_{1 \leq j \leq N}^{j \neq n}} a_{k} q_{k}, \forall a_{k} \in \mathbb{Q}\}$. According to the above definition, $\mathcal{Q}_{[-1,1]}^{n}(\{\hat{\vartheta}_{j}\}_{1 \leq j \leq N}^{j \neq n})$ is also a countable set with its Lebesgue measure being zero. It is worth noting that $\bm{\vartheta} \in \mathcal{J}$ if and only if $\hat{\vartheta}_{n} \notin \mathcal{Q}_{[-1,1]}^{n}(\{\hat{\vartheta}_{j}\}_{1 \leq j \leq N}^{j \neq n})$ holds for $\forall n=1,2,\cdots,N$. Thus, we have 
\begin{equation}
	\begin{aligned}		
		&\int_{\bm{\vartheta} \in \mathcal{J}} f_{\bm{\Theta}}\left\{\bm{\vartheta}\right\} \dif \bm{\vartheta}  \\
		\geq& 1-\sum \limits_{n=1}^{N} \idotsint_{-1}^{1} 
		\Pr\left\{ \Theta_{n} \in \mathcal{Q}_{[-1,1]}^{n}(\{\hat{\vartheta}_{j}\}_{1 \leq j \leq N}^{j \neq n}) \right\}\\
		&\times \prod_{j=1, j \neq n}^{N} f_{\Theta_{j}}(\hat{\vartheta}_{j}) ~\dif \hat{\vartheta}_{1} \cdots \dif \hat{\vartheta}_{n-1} \dif \hat{\vartheta}_{n+1} \cdots \dif \hat{\vartheta}_{N}\\
		\geq&1-\sum \limits_{n=1}^{N} \idotsint_{-1}^{1} 
		\varrho \times \mathcal{L}\left( \mathcal{Q}_{[-1,1]}^{n}(\{\hat{\vartheta}_{j}\}_{1 \leq j \leq N}^{j \neq n}) \right)\\
		&\times \prod_{j=1, k \neq n}^{N} f_{\Theta_{j}}(\hat{\vartheta}_{j}) ~\dif \hat{\vartheta}_{1} \cdots \dif \hat{\vartheta}_{n-1} \dif \hat{\vartheta}_{n+1} \cdots \dif \hat{\vartheta}_{N}\\
		=&1,
	\end{aligned}
\end{equation}
which thus completes the proof.

\section{Proof of Theorem \ref{Theo_CIR}}\label{Appen_CIR}
To prove Theorem \ref{Theo_CIR}, we introduce the concept of uniformly distributed sequence \cite{kuipers2012uniform}. Specifically, an $N$-dimensional sequence of real vectors $\{\mathbf{u}_{k} \in \mathbb{R}^{N}\}_{1 \leq k \leq K}$ is uniformly distributed modulo 1 if for all $N$-dimensional intervals $\mathcal{B}=\prod_{n=1}^{N} [a_{n}, b_{n}] \subseteq [0, 1)^{N}$, the following equation always holds,
\begin{equation}\label{eq_uniform}
	\begin{aligned}
		&\lim \limits_{K\rightarrow +\infty} \frac{\big{|}\left\{k ~|~(\mathbf{u}_{k} \mod 1) \in \mathcal{B}, 1 \leq k \leq K\right\}\big{|}}{K} \\
		&= \prod_{n=1}^{N} (b_{n}-a_{n}).
	\end{aligned}
\end{equation}

It was shown in \cite[Chap 1]{granville2007equidistribution} that if all elements in an $N$-dimensional real-valued vector $\mathbf{u}$ are linearly independent over $\mathbb{Q}$, the sequence $\{\mathbf{u}_{k}=k\mathbf{u}\}_{k \in \mathbb{N}}$ is uniformly distributed modulo 1. Following this conclusion, we consider vector $\mathbf{u}\triangleq \frac{1}{\lambda} \bm{\vartheta} = \frac{1}{\lambda}[\hat{\vartheta}_{1}^{1}, \cdots, \hat{\vartheta}_{1}^{L_{1}}, \cdots, \hat{\vartheta}_{T}^{1}, \cdots, \hat{\vartheta}_{T}^{L_{T}}]^{\mathrm{T}}$, of which the dimension is given by $N=\sum_{\tau=1}^{T} L_{\tau}$. Note that the elements in $\bm{\vartheta}$ (as well as $\mathbf{u}$) satisfy the LIA condition, i.e., they are linearly independent over $\mathbb{Q}$. Thus, the defined vector sequence $\{\mathbf{u}_{k}=k\mathbf{u}\}_{k \in \mathbb{N}}$ is uniformly distributed modulo 1. Then, for any given small positive $\delta \ll 1$ and $\nu_{\tau}$, $ 1\leq \tau \leq T$, we define an $N$-dimensional interval
\begin{equation}\label{eq_box}
	\begin{aligned}
		\mathcal{B}&\triangleq \prod_{n=1}^{N} \mathcal{B}_{n}
		\triangleq \prod_{\tau=1}^{T} \prod_{\ell=1}^{L_{\tau}} \left[a_{\tau}^{\ell}, a_{\tau}^{\ell}+\frac{\delta}{\Delta}\right],
	\end{aligned}
\end{equation}
with $\Delta \triangleq \frac{1}{2\pi}\sum_{\tau=1}^{T} \|\mathbf{b}_{\tau}\|_{1}$ and $a_{\tau}^{\ell} \triangleq \nu_{\tau}-\frac{\angle b_{\tau}^{\ell}}{2\pi} \mod 1$, $1 \leq \ell \leq L_{\tau}$, $1 \leq \tau \leq T$.
Substituting \eqref{eq_box} into \eqref{eq_uniform}, we have 
\begin{equation}\label{eq_boxsize}
	\begin{aligned}
		&\lim \limits_{K\rightarrow +\infty} \frac{\big{|}\left\{k ~|~(\mathbf{u}_{k} \mod 1) \in \mathcal{B}, 1 \leq k \leq K\right\}\big{|}}{K} \\
		&= \left(\frac{\delta}{\Delta}\right)^{N} > 0.
	\end{aligned}
\end{equation}
It indicates that there always exists an $k \in \mathbb{N}$ ensuring 
\begin{equation}\label{eq_phase_adj}
	a_{\tau}^{\ell} \leq \frac{k}{\lambda} \hat{\vartheta}_{\tau}^{\ell} \leq a_{\tau}^{\ell}+\frac{\delta}{\Delta}, ~1 \leq \ell \leq L_{\tau}, ~1 \leq \tau \leq T.
\end{equation}
By setting the MA positioning vectors as $\mathbf{t}=[k, 0, 0]^{\mathrm{T}}$ and $\mathbf{r}=[0, 0, 0]^{\mathrm{T}}$, we have
\begin{equation}\label{eq_CIRgap}
	\begin{aligned}
		&\left|\|\mathbf{b}_{\tau}\|_{1} e^{j2\pi\nu_{\tau}} - h_{\tau}(\mathbf{t}, \mathbf{r})\right| \\
		\overset{(a)}{=}&\left|\|\mathbf{b}_{\tau}\|_{1} e^{j2\pi\nu_{\tau}} - \sum \limits_{\ell=1}^{L_{\tau}} b_{\tau}^{\ell} e^{j \frac{2\pi}{\lambda}k \hat{\vartheta}_{\tau}^{\ell}}\right|\\
		=&\left|\sum \limits_{\ell=1}^{L_{\tau}} |b_{\tau}^{\ell}| - \sum \limits_{\ell=1}^{L_{\tau}} b_{\tau}^{\ell} e^{j2\pi (\frac{k}{\lambda} \hat{\vartheta}_{\tau}^{\ell}-\nu_{\tau})}\right|\\
		\overset{(b)}{\leq} & \sum \limits_{\ell=1}^{L_{\tau}}\left| |b_{\tau}^{\ell}| - b_{\tau}^{\ell} e^{j2\pi (\frac{k}{\lambda} \hat{\vartheta}_{\tau}^{\ell}-\nu_{\tau})}\right|\\
		= & \sum \limits_{\ell=1}^{L_{\tau}} |b_{\tau}^{\ell}| \times \left|1-e^{j2\pi (\frac{k}{\lambda} \hat{\vartheta}_{\tau}^{\ell}-\nu_{\tau}+\frac{\angle b_{\tau}^{\ell}}{2\pi})}\right|\\
		\overset{(c)}{=} & \sum \limits_{\ell=1}^{L_{\tau}} |b_{\tau}^{\ell}| \times \left|1-e^{j2\pi (\frac{k}{\lambda} \hat{\vartheta}_{\tau}^{\ell}-a_{\tau}^{\ell})}\right|\\
		\overset{(d)}{\leq} & \sum \limits_{\ell=1}^{L_{\tau}} |b_{\tau}^{\ell}| \times 2\pi \frac{\delta}{\Delta}
		\overset{(e)}{\leq}  \delta, ~1 \leq \tau \leq T,
	\end{aligned}
\end{equation}
where step $(a)$ holds by substituting $\mathbf{t}=[k, 0, 0]^{\mathrm{T}}$ and $\mathbf{r}=[0, 0, 0]^{\mathrm{T}}$ into $h_{\tau}(\mathbf{t}, \mathbf{r})$ in \eqref{eq_CIR}; step $(b)$ holds based on the triangle inequality; step $(c)$ holds according to the definition of $a_{\tau}^{\ell} = \nu_{\tau}-\frac{\angle b_{\tau}^{\ell}}{2\pi} \mod 1$; step $(d)$ holds because of \eqref{eq_phase_adj} and the fact of $|1-e^{jx}| \leq X$ for $0 \leq x \leq X$; and step $(e)$ holds due to $\Delta = \frac{1}{2\pi}\sum_{\tau=1}^{T} \|\mathbf{b}_{\tau}\|_{1}$. This thus completes the proof.


\section{Proof of Theorem \ref{Theo_rate}}\label{Appen_rate}
According to \eqref{eq_CIR} and Theorem \ref{Theo_CIR}, the total power of the CIR vector is upper-bounded by 
\begin{equation}\label{eq_CIR_bound}
	\left\|\mathbf{h}(\mathbf{t},\mathbf{r})\right\|_{2}^{2} \leq \sum \limits_{\tau=1}^{T} \|\mathbf{b}_{\tau}\|_{1}^{2} = G.
\end{equation}
According to \eqref{eq_CFR}, the total power of the CFR vector is upper-bounded by
\begin{equation}
	\begin{aligned}
		&\left\|\mathbf{c}(\mathbf{t},\mathbf{r})\right\|_{2}^{2} = \left\|\mathbf{D}\mathbf{h}(\mathbf{t},\mathbf{r})\right\|_{2}^{2} =\mathbf{h}(\mathbf{t},\mathbf{r})^{\mathrm{H}}\mathbf{D}^{\mathrm{H}}\mathbf{D}\mathbf{h}(\mathbf{t},\mathbf{r}) \\
		&=M \mathbf{h}(\mathbf{t},\mathbf{r})^{\mathrm{H}}\mathbf{I}_{M}\mathbf{h}(\mathbf{t},\mathbf{r})=M\left\|\mathbf{h}(\mathbf{t},\mathbf{r})\right\|_{2}^{2} \leq MG.
	\end{aligned}
\end{equation}
Note that according to Theorem \ref{Theo_CIR}, the upper bound on the total channel power gain, $G$, can be asymptotically approached by optimizing the positions of the Tx-MA and Rx-MA. Thus, the optimal objective value for problem \eqref{eq_problem} is upper-bounded by the following relaxed problem:
\begin{subequations}\label{eq_problem_relax}
	\begin{align}
		\mathop{\max}\limits_{\mathbf{v},\mathbf{p}}~~~
		&\tilde{R}(\mathbf{v},\mathbf{p}) \label{eq_problem_relax_a}\\
		\mathrm{s.t.}~~~~ &p_{m} \geq 0,~1 \leq m \leq M, \label{eq_problem_relax_b}\\
		&\sum \limits_{m=1}^{M} p_{m} \leq P, \label{eq_problem_relax_c}\\
		&v_{m} \geq 0,~1 \leq m \leq M, \label{eq_problem_relax_d}\\
		&\sum \limits_{m=1}^{M} v_{m} \leq MG, \label{eq_problem_relax_e}
	\end{align}
\end{subequations}
where $\mathbf{v}=[v_{1}, v_{2}, \cdots, v_{M}]^{\mathrm{T}}$ is an $M$-dimensional vector representing the channel power gain over all the subcarriers and $\tilde{R}(\mathbf{v},\mathbf{p}) \triangleq \frac{1}{M+M_{\mathrm{CP}}} \sum_{m=1}^{M}\log_{2} \left(1+\frac{v_{m}p_{m}}{\sigma^{2}}\right)$ is the OFDM achievable rate.

Let $p_{m}^{\star}$ and $v_{m}^{\star}$ denote the optimal solution for transmit power allocation and channel power gain realization over the $m$-th OFDM subcarrier, respectively. It is easy to verify that problem \eqref{eq_problem_relax} is convex w.r.t. $\mathbf{p}$ and its optimal solution satisfies the water-filling criterion, i.e.,
\begin{equation}\label{eq_power_allo}
	p_{m}^{\star} = \max \left\{\mu_{1} - \frac{\sigma^{2}}{v_{m}^{\star}}, 0\right\},~1 \leq m \leq M,
\end{equation}
where $\mu_{1}$ should be selected to guarantee $\sum_{m=1}^{M} p_{m}^{\star} = P$. In the high-SNR regime, we always have $\mu_{1} - \frac{\sigma^{2}}{v_{m}^{\star}} >0$, $1 \leq m \leq M$, and thus \eqref{eq_power_allo} can be simplified as
\begin{equation}\label{eq_power_allo2}
	p_{m}^{\star} = \mu_{1} - \frac{\sigma^{2}}{v_{m}^{\star}} = \frac{P}{M} + \frac{1}{M} \sum \limits_{n=1}^{M} \frac{\sigma^{2}}{v_{n}^{\star}} - \frac{\sigma^{2}}{v_{m}^{\star}},~1 \leq m \leq M.
\end{equation}

Next, we show the proof by contradiction. Specifically, we assume that for the optimal $\mathbf{v}^{\star}$, there exist two elements satisfying $v_{m_{1}}^{\star} > v_{m_{2}}^{\star}$, $1 \leq m_{1} \neq m_{2} \leq M$. Then, we define a new solution $\mathbf{v}^{\circ}$ for problem \eqref{eq_problem_relax} as 
\begin{equation}\label{eq_channel_power_allo_new}
	v_{m}^{\circ} = \left\{ \begin{aligned}
		&v_{m}^{\star},~1 \leq m \leq M,~m \neq m_{1}, m_{2},\\
		&(v_{m_{1}}^{\star}+v_{m_{2}}^{\star})/2,~m=m_{1}, m_{2},
	\end{aligned}\right.
\end{equation}
which satisfy constraints \eqref{eq_problem_relax_d} and \eqref{eq_problem_relax_e}. Accordingly, we can derive the difference of achievable rates under solutions $(\mathbf{v}^{\circ},\mathbf{p}^{\star})$ and $(\mathbf{v}^{\star},\mathbf{p}^{\star})$ as
\begin{equation}\label{eq_rate_diff}
	\begin{aligned}
		& \left(M+M_{\mathrm{CP}}\right) \times \left(\tilde{R}(\mathbf{v}^{\circ},\mathbf{p}^{\star}) - \tilde{R}(\mathbf{v}^{\star},\mathbf{p}^{\star})\right)\\
		=&\log_{2} \left(1+\frac{v_{m_{1}}^{\circ}p_{m_{1}}^{\circ}}{\sigma^{2}}\right) 
		+ \log_{2} \left(1+\frac{v_{m_{2}}^{\circ}p_{m_{2}}^{\circ}}{\sigma^{2}}\right)\\
		&- \log_{2} \left(1+\frac{v_{m_{1}}^{\star}p_{m_{1}}^{\star}}{\sigma^{2}}\right) 
		- \log_{2} \left(1+\frac{v_{m_{2}}^{\star}p_{m_{2}}^{\star}}{\sigma^{2}}\right)\\
		\triangleq &\log_{2} \left(1+\frac{\xi}{\left(1+\frac{v_{m_{1}}^{\star}p_{m_{1}}^{\star}}{\sigma^{2}}\right) \times \left(1+\frac{v_{m_{2}}^{\star}p_{m_{2}}^{\star}}{\sigma^{2}}\right)}\right),
	\end{aligned}
\end{equation}
with 
\begin{equation}\label{eq_rate_diff_xi}
	\begin{aligned}
		\xi = &\left(1+\frac{v_{m_{1}}^{\circ}p_{m_{1}}^{\circ}}{\sigma^{2}}\right) \times \left(1+\frac{v_{m_{2}}^{\circ}p_{m_{2}}^{\circ}}{\sigma^{2}}\right)\\
		&-\left(1+\frac{v_{m_{1}}^{\star}p_{m_{1}}^{\star}}{\sigma^{2}}\right) \times \left(1+\frac{v_{m_{2}}^{\star}p_{m_{2}}^{\star}}{\sigma^{2}}\right)\\
		=&\frac{v_{m_{1}}^{\star}+v_{m_{2}}^{\star}}{2\sigma^{2}}p_{m_{1}}^{\star} + \frac{v_{m_{1}}^{\star}+v_{m_{2}}^{\star}}{2\sigma^{2}}p_{m_{2}}^{\star}
		- \frac{v_{m_{1}}^{\star}p_{m_{1}}^{\star}}{\sigma^{2}} - \frac{v_{m_{2}}^{\star}p_{m_{2}}^{\star}}{\sigma^{2}}\\
		&+ \left(\frac{v_{m_{1}}^{\star}+v_{m_{2}}^{\star}}{2\sigma^{2}}\right)^{2}p_{m_{1}}^{\star}p_{m_{2}}^{\star}
		- \frac{v_{m_{1}}^{\star}v_{m_{2}}^{\star}p_{m_{1}}^{\star}p_{m_{2}}^{\star}}{\sigma^{4}} \\
		=&\frac{v_{m_{1}}^{\star}-v_{m_{2}}^{\star}}{2\sigma^{2}} \left(p_{m_{2}}^{\star}-p_{m_{1}}^{\star}\right) 
		+ \left(\frac{v_{m_{1}}^{\star}-v_{m_{2}}^{\star}}{2\sigma^{2}}\right)^{2} p_{m_{1}}^{\star}p_{m_{2}}^{\star}\\
		\overset{(e)}{=}& \frac{v_{m_{1}}^{\star}-v_{m_{2}}^{\star}}{2\sigma^{2}} \left(\frac{\sigma^{2}}{v_{m_{1}}^{\star}}-\frac{\sigma^{2}}{v_{m_{2}}^{\star}}\right) 
		+ \left(\frac{v_{m_{1}}^{\star}-v_{m_{2}}^{\star}}{2\sigma^{2}}\right)^{2} p_{m_{1}}^{\star}p_{m_{2}}^{\star}\\
		=&\left(\frac{v_{m_{1}}^{\star}-v_{m_{2}}^{\star}}{2\sigma^{2}}\right)^{2} \left(p_{m_{1}}^{\star}p_{m_{2}}^{\star}-\frac{2\sigma^{4}}{v_{m_{1}}^{\star}v_{m_{2}}^{\star}}\right),
	\end{aligned}
\end{equation}
where step $(e)$ is based on the definition of $p_{m}^{\star}$ in \eqref{eq_power_allo2}. Note that in the high-SNR regime, we have $p_{m}^{\star} v_{m}^{\star} \gg \sigma^{2}$, $m=m_{1},m_{2}$. Thus, we have $p_{m_{1}}^{\star}p_{m_{2}}^{\star}-\frac{2\sigma^{4}}{v_{m_{1}}^{\star}v_{m_{2}}^{\star}} > 0$, which indicates $\xi > 0$. Furthermore, according to \eqref{eq_rate_diff}, we have $\tilde{R}(\mathbf{v}^{\circ},\mathbf{p}^{\star}) - \tilde{R}(\mathbf{v}^{\star},\mathbf{p}^{\star}) >0$, which contradicts to the fact that $(\mathbf{v}^{\star},\mathbf{p}^{\star})$ is the optimal solution for problem \eqref{eq_problem_relax}. 

As such, we can conclude that $v_{m_{1}}^{\star} = v_{m_{2}}^{\star}$ always holds for any $1 \leq m_{1}, m_{2} \leq M$. To maximize $\tilde{R}(\mathbf{v},\mathbf{p})$ and satisfy constraint \eqref{eq_problem_relax_e}, the optimal channel power gain realization for problem \eqref{eq_problem_relax} is thus given by $v_{m}^{\star} = G$, $1 \leq m \leq M$. Substituting $v_{m}^{\star} = G$ into \eqref{eq_power_allo2}, the optimal power allocation can be simplified as $p_{m}^{\star}=\frac{P}{M}$, $1 \leq m \leq M$. Substituting the optimal $\mathbf{v}^{\star}$ and $\mathbf{p}^{\star}$ into \eqref{eq_problem_relax_a}, we obtain $\tilde{R}(\mathbf{v}^{\star},\mathbf{p}^{\star}) = \frac{1}{M+M_{\mathrm{CP}}} \sum_{m=1}^{M}\log_{2} \left(1+\frac{GP}{M\sigma^{2}}\right) = \frac{M}{M+M_{\mathrm{CP}}} \log_{2} \left(1+\frac{GP}{M\sigma^{2}}\right)$, which is an upper bound on the OFDM achievable rate in the high-SNR regime. This thus completes the proof.

\bibliographystyle{IEEEtran} 
\bibliography{IEEEabrv,ref_zhu}

\end{document}